\documentclass[twoside,12pt]{article} 
 
 \usepackage{natbib}
 \usepackage{hyperref}
 \hypersetup{
 	colorlinks=true,
 	linkcolor=blue,
 	filecolor=magenta,      
 	urlcolor=cyan,
 	citecolor=blue,
 }
 \usepackage{amssymb,bbm,authblk}
 \usepackage{amsmath}
 \usepackage{amsthm}
 \usepackage{graphicx}
 \usepackage{algorithm}
 \usepackage{algpseudocode}
 \usepackage{geometry}
\RequirePackage{amsfonts}
\usepackage{verbatim}
\usepackage{subfigure}
\usepackage{multirow}
\usepackage{bm}
\usepackage{enumitem}
\usepackage{tablefootnote}
\usepackage{xr} 

 \newtheorem{theorem}{Theorem}[section]
 \newtheorem{lemma}[theorem]{Lemma}
 
 \newtheorem{assumption}[theorem]{Assumption}
 \newtheorem{definition}[theorem]{Definition}
 \newtheorem{example}{Example}
 
 \newtheorem{remark}{Remark}
 
 \newcommand{\bb}{\mathbb}
 \newcommand{\RNum}[1]{\uppercase\expandafter{\romannumeral #1\relax}}
\date{\vspace{-5ex}}

\begin{document}
	\title{Random Interval Distillation for Detection of Change-Points in Markov Chain Bernoulli Networks}
	\author[1]{Xinyuan Fan}
	\author[1]{Weichi Wu}
	\affil[1]{\small Department of Statistics and Data Science, Tsinghua University}
	
	\maketitle

	\begin{abstract}%
We propose a new and generic approach for detecting multiple change-points in dynamic networks with Markov formation, termed random interval distillation (RID). By collecting random intervals with sufficient strength of signals and reassembling them into a sequence of informative short intervals, together with sparse universal singular value thresholding, our new approach can achieve nearly minimax optimality as their independent counterparts for both detection and localization bounds in low-rank networks without any prior knowledge about minimal spacing, which is unlike many previous methods. In particular,  motivated by a recent nonasymptotic bound \citep{bandeira2016sharp}, our method uses the operator norm of CUSUMs of the adjacency matrices, and achieves the aforementioned optimality without sample splitting as required by the previous method. For practical applications, we introduce a clustering-based and data-driven procedure to determine the optimal threshold for signal strength,  utilizing the connection between RID and clustering. We examine the effectiveness and usefulness of our methodology via simulations and a real data example. 
	\end{abstract}
	
	\

	\section{Introduction}
Change-point detection methods have significant prospects for applications in various fields. With advancements in data collection technology, an increasing amount of dependent network data is being observed, such as in \cite{mersch2013tracking, gemmetto2014mitigation, stehle2011high}. However, most existing multiple change-point detection methods are developed with theoretical guarantees primarily for independent  Euclidean data, or require strict stationarity between changes. Recently, \cite{wang2021optimal} and \cite{zambon2019change} discussed detecting multiple changes in network sequences where networks are independent of each other at different time points. \cite{padilla2019change} investigated the multiple changes in latent structure for specific dynamic networks related to the random dot product graphs (RDPG). Therefore, these methods are not directly applicable to investigating multiple structural changes in general dependent dynamic networks. To bridge this gap, we aim to develop a novel and general approach for multiple change-point detection in dynamic networks.

There is a vast amount of literature that considers estimation or inference in dynamic networks. For instance, \cite{10.1214/18-AOS1751} proposed a dynamic stochastic block model and estimated the tensor of connection probabilities using the vectorization technique. \cite{zhao2019estimation} proposed a dynamic network model with a Markovian structure and constructed likelihood estimators for the model parameters, deriving their asymptotic properties. \cite{jiang2020autoregressive} built a strictly stationary autoregressive network model and investigated the corresponding estimation and inference problems.  
In this paper, we consider dynamic networks with a general Markov formulation. This approach includes existing work \citep{jiang2020autoregressive, zhao2019estimation} as special cases and allows heterogeneity between two change-points, aligning with patterns observed in real data. 
In particular, we assume that the connection probabilities $\mathbb E(A(t))$ where $A(t)$ is the adjacency matrix at time $t$, may change at unknown time points. We aim to develop methods that effectively detect all possible change-points accurately under minimal restrictions, admitting a Markov structure within $A_{ij}(t)$. With effective change-point detection, one can apply static network estimation methods to the means of the adjacency matrices between two consecutive change-point estimators to explore further properties, such as clustering \citep{pmlr-v37-hanb15}.

Previous multiple change point detection methods for network sequences, such as \cite{wang2021optimal} and \cite{padilla2019change}, require the knowledge of minimal spacing to ensure theoretically guaranteed good performance. However, the minimal spacing is usually unknown in practice. One reason for this limitation is that these methods are based on wild binary segmentation (WBS). To overcome this challenge, we propose a new Random Interval Distillation (RID) technique in Section~\ref{sec:Change-point estimation} for detecting multiple change-points that does not require prior knowledge of minimal spacing, and can be of separate interest in the field of structural change analysis. Moreover,  we propose in Section~\ref{subsec:parac} a data-driven threshold selection method adaptive to the complex dependencies. This offers significant practical guidance, which is notably absent in most of the aforementioned studies. In the literature of change-point detection, there are primarily three branches: loss and penalization based methods, moving sum (MOSUM, originated from \cite{antoch2000change}) based methods, and binary segmentation based methods. The first class of methods typically employs a loss function, often related to likelihood or quasi-likelihood, along with some regularization to estimate the best piecewise constant signals and, consequently, detect changes. For strictly stationary autoregressive networks, \cite{jiang2020autoregressive} used the likelihood approach to detect changes in transition probabilities under the at-most-one-change (AMOC) setting. However, it is difficult to generalize these methods to dependent dynamic networks with heterogeneous transition probabilities. On the other hand, detecting change-points in dependent sequences using the second class of methods, MOSUM-based approaches, is challenging due to the necessity of selecting an appropriate window size or bandwidth. This task is particularly difficult for network sequences with temporal dependence and heterogeneity. 
 
The third category of multiple change-point detection methods is Binary Segmentation (BS), as referenced in \citet{vostrikova1981detecting, venkatraman1992consistency}, and the seminal Wild Binary Segmentation (WBS) algorithm \citep{fryzlewicz2014wild}. In \cite{wang2021optimal}, independent Bernoulli networks were considered. They proposed a method that involves splitting the sample, maximizing the inner product, and applying WBS to detect change points. They also incorporated a refinement technique using Universal Singular Value Thresholding (USVT) \citep{chatterjee2015matrix} for low-rank networks. However, practical guidelines for adjusting thresholds in the WBS step were limited. Even for the univariate sequence, it is well known that selecting suitable thresholds is challenging when the data is neither independent nor strictly stationary. Additionally, as mentioned before, to achieve a nearly minimax optimal detection bound, both the methods outlined in \cite{fryzlewicz2014wild} and other WBS-based methods (e.g., \cite{wang2021optimal, padilla2019change}) require prior knowledge of minimal spacing when the number of change-points is unbounded.

In this paper, we introduce a novel procedure for change-point detection in dynamic networks with a Markov formulation that is distinct from the existing three classes of methods. Our approach, called Random Interval Distillation (RID), effectively get RID of the influence of structural changes. 
 Combined with Sparse Universal Singular Value Thresholding (SUSVT) proposed in Section \ref{newsec:K}, our method achieves the nearly minimax detection bound and minimax localization bound, up to logarithmic terms, as their independent counterparts if the network has certain low-rank structure. We sample random intervals and calculate the operator norms of cumulative sums  (CUSUMs) of adjacency matrices as signals over intervals. This differs from \cite{wang2021optimal} where adjacency matrices were vectorized and treated as high-dimensional vectors, and leads us to a nearly minimax optimal detection bound up to a logarithmic term in low-rank networks without splitting the samples. 
  
 A major difference between RID and the WBS-based methods is that we do not sequentially separate the random intervals. Instead, we process all random intervals together to create a new sequence of informative intervals. These new intervals, which may not correspond to any of the original intervals, each contain exactly one change-point and are typically very short with high probability, as illustrated in Figure~\ref{fig:intro not}. We term this process `distillation' because the new sequence retains the full information of changes from the original noisy intervals which may contain none or multiple change-points. RID can handle heterogeneity and does not require prior knowledge about minimal spacing when the number of changes is unbounded, which is highly advantageous in practical applications. Additionally, we propose a new data-driven threshold using a clustering algorithm, supported by Theorem~\ref{the:gmm} that guarantees the existence of clustering boundaries for informative and non-informative intervals. It is worth emphasizing that popular seeded intervals from \cite{kovacs2020seeded} provide insufficient information to distinguish between intervals with and without change points for identifying clustering boundaries, as shown in Figure~\ref{parafig:8}.  To the best of the authors' knowledge, RID is the first method that connects change-point detection to clustering with theoretical justification. For refinement under the existence of the low-rank structure, based on the constructed informative intervals and USVT, we develop the SUSVT  for localization tailored to the dependence in Markov chain Bernoulli networks, achieving a nearly minimax optimal rate as its independent counterpart.
 We summarize our contributions as follows:

\begin{enumerate}
    \item \textbf{Flexibility of the Markov Chain Bernoulli Network model:}
The Markov chain Bernoulli network considered in this paper is flexible enough to include existing dynamic network models such as \cite{wang2021optimal}, \cite{jiang2020autoregressive}, and \cite{zhao2019estimation}. The Markov chain between two change-points is allowed to be heterogeneous with time-varying transition matrices, where the number of nuisance parameters tends to infinity as 
$T$ diverges. Furthermore, the heterogeneous Markov chain has also been used to model non-exchangeability for a single observed network, as discussed in \cite{wu2020tractably}.
\item \textbf{Introduction of RID and SUSVT:}
We propose RID, an alternative method for multiple change-point detection that uses random intervals differently from Wild Binary Segmentation, to eliminate the need for prior knowledge of minimal spacing. Moreover, random intervals facilitate the development of an adaptive threshold selection method based on clustering principles.
We further introduce SUSVT for refinement, so that our estimator of changes achieves both nearly minimax detection bound and optimal localization rate for networks with low-rank marginal connectivity probability matrices. Both RID and SUSVT are applicable to the above mentioned flexible Markov chain Bernoulli network models allowing heterogeneity between change-points. RID employs the operator norm instead of the conventional Frobenius norm to aggregate local network cumulative sums, thereby increasing the effective sample size and eliminating the need for sample splitting required by \cite{wang2021optimal}.
\item \textbf{Theoretical contributions:}
To address the theoretical challenges posed by nonstationarity and the heterogeneous Markov structure, we introduce recent results in matrix norms \citep{bandeira2016sharp} and a recent concentration inequality for Markov chains \citep{paulin2015concentration}. These techniques have the potential to be applied to the inference of statistical models for general nonstationary dependent matrix data. Moreover, we provide a theoretical justification for the clustering boundary used in threshold selection, thereby linking change-point detection with clustering.
\end{enumerate}

The rest of the paper is structured as follows. Section~\ref{sec:Problem setup} introduces the Markov chain Bernoulli network and the general settings for change-point detection. Section~\ref{sec:Change-point estimation} presents the algorithms and their theoretical properties. In Section~\ref{subsec:parac}, we introduce a fully data-driven method based on clustering for selecting the threshold. In Section~\ref{sec:sim}, we evaluate the performance of our method through numerical studies, while in Section~\ref{sec:real data}, we  analyze a real-world example, identifying meaningful changes in contact patterns among individuals in a primary school. In Section \ref{sec:conclusion} we  provide the conclusion and possible extensions for future work. We provide the proofs of  Theorem~\ref{th:K}, Theorem~\ref{th:Theoretical result of local refinement}, Theorem~\ref{th:network refinement sbm}, Theorem~\ref{the:gmm} and related lemmas, along with descriptive analyses and figures for the real data example, are included in the appendix.

	\subsection{Notations}  
	For a set $S$, let $|S|$ be the number of elements in $S$. 
	For a real number $x$, $\lfloor x\rfloor$ represents the largest integer that is less than or equal to $x$, and $sign(x)=I(x>0)-I(x<0)$ where $I(\cdot)$ is the indicator function. For two positive real numbers $a,b$, write $a\vee b=max(a,b), a\wedge b=min(a,b)$. For two matrices $A$ and $B$, define the inner product $<A,B>=\sum_{i=1}^n\sum_{j=1}^n A_{ij}B_{ij}$ where $A_{ij}, B_{ij}$ are the elements in $i_{th}$ row and $j_{th}$ column of $A$ and $B$, respectively. 
	For a matrix $A$, let $\text{rank}(A)$ be the rank of $A$, 
	$||A||_{F}$ be the Frobenius norm, and $||A||_{op}$ be the operator norm (i.e., the maximum absolute eigenvalue). 
	For two sequences of positive real numbers $a_n,b_n, a_n=o(b_n)$ if $a_n/b_n\to 0$ as $n\to \infty$, $a_n=O(b_n), a_n\lesssim b_n$ if there exists two positive constants $N$ and $C$ such that $a_n/b_n\le C$ when $n>N$, and $a_n=\Omega(b_n)$ if there exists three positive constants $N, C_1, C_2$ such that $C_1\le a_n/b_n\le C_2$ when $n>N$. Let $C_i, c_i, i=1,2,\cdots$ denote absolute constants.

	\section{Preliminary}
	\label{sec:Problem setup}

	\subsection{Markov chain Bernoulli network}
	
 In practice, the probabilities of connections between nodes often exhibit temporal dependencies in fields. For example, in the business field,  for trade data (e.g., \cite{jiang2020autoregressive}), the annual international trades between countries show temporal dependence.  In ecology, for ant societies data (\cite{mersch2013tracking}), the contact patterns between most of the ant pairs are dependent since they follow rules that incorporate local stimuli from the environment and previous interactions. In social science,  \cite{stehle2011high} identified temporal dependence of social networks across different times, which is further confirmed by our data analysis in Section~\ref{sec:real data}. To allow general temporal dependence between connections, we define the Markov chain Bernoulli network as follows.

\begin{definition}[Markov chain Bernoulli network]
		\label{def:Bernoulli network Markov chain}
		Let $\{A(t)\}_{t=1}^T$ be a series of symmetric adjacency matrices with size $n$. For $1\le i\le j\le n, 1\le t\le T$, the entity $A_{ij}(t)\in\{0,1\}$ are Bernoulli random variables and  $A_{ij}(t)=1$ if the $i_{th}$ and $j_{th}$ nodes  are connected at time $t$.   $\{A(t)\}_{t=1}^T$ is a Markov chain Bernoulli network  if 
		\begin{enumerate}
			\item For any fixed $1\le t\le T$, $A_{ij}(t)$ is independent of $A_{i'j'}(t)$ for all $(i,j)\neq (i',j'), i\le j,i'\le j'$.
			\item For any fixed $1\le i\le j\le n, \{A_{ij}(t)\}_{t=1}^T$ is a Markov chain.
		\end{enumerate}
	\end{definition}
	
	The above definition is mainly for undirected networks, but our results can be easily extended to directed networks. Furthermore, our results are easy to extend to weighted graphs, though we focus on Bernoulli networks in this paper for simplicity. When the Markov chain is homogeneous, detecting changes in connection probabilities is equivalent to that in transition probabilities, and  \cite{jiang2020autoregressive} studied the latter for cases that contain exactly one change-point, requiring $n$ to be divergent.  In this paper, we allow the number of change-points to be fixed or divergent and further allow Markov chains to be heterogeneous when detecting whether $\bb{E}A(t)$ stays constant across time. We remark that the independence condition (Definition~\ref{def:Bernoulli network Markov chain} (1)) is a technical condition as it avoids the challenge of controlling the operator norm of a matrix with dependent entries. It is possible to weaken this condition to conditional independence, which we leave for future work.  
	Our Markov chain Bernoulli network model is very general, including following examples.
	\begin{example}[Independent Bernoulli networks]
 \label{ex:1}
		When $A_{ij}(t)$ is independent of $A_{i'j'}(t')$ for any $(i,j,t)\neq (i',j',t')$, $\{A(t)\}_{t=1}^T$ is a sequence of  independent Bernoulli networks. Notice that for any fixed $t_1$, $A(t_1)$ covers a range of networks, including the stochastic block model \citep{holland1983stochastic}, RDPG \citep{young2007random}, to name a few. 
  \cite{wang2021optimal} proposed an optimal change-point detection method for the above $\{ A(t)\}_{t=1}^T$, however, their method is not applicable when Markovian dependencies between edges exist within the network.
	\end{example}
	\begin{example}[Autoregressive networks]
		
  \cite{jiang2020autoregressive} proposed autoregressive networks defined as \(A_{ij}(t) = A_{ij}(t-1)I(\epsilon_{ij}=0) + I(\epsilon_{ij}=1)\), where \(\epsilon_{ij}\) represents independent innovations with \(\mathbb{P}(\epsilon_{ij}=1) = \alpha_{ij}\), \(\mathbb{P}(\epsilon_{ij}=-1) = \beta_{ij}\), and \(\mathbb{P}(\epsilon_{ij}=0) = 1 - \alpha_{ij} - \beta_{ij}\), subject to the constraint \(\alpha_{ij} + \beta_{ij} \leq 1\). They proposed a method for estimating the parameters \(\alpha_{ij}\) and \(\beta_{ij}\), assuming the sample is drawn from a stationary process. They also studied the change-point of the above trade data under the assumption of at-most-one-change. Notice that their model is equivalent to \(\mathbb{P}(A_{ij}(t)=1 | A_{ij}(t-1)=0) = \alpha_{ij}\) and \(\mathbb{P}(A_{ij}(t)=0 | A_{ij}(t-1)=1) = \beta_{ij}\) for \(\alpha_{ij} + \beta_{ij} \leq 1\), and thus is included in our model as a homogeneous case. Our model is more flexible in the sense that it allows for heterogeneity by permitting the parameters \(\alpha_{ij}\) and \(\beta_{ij}\) to depend on \(t\), and does not require the restriction \(\alpha_{ij} + \beta_{ij} \leq 1\). 
For this model, we investigate the multiple change-point detection problem where the number of changes is allowed to diverge.
	\end{example}
	\begin{example}[Model in \cite{zhao2019estimation}]
	
\cite{zhao2019estimation} defined their dynamic random graph model as \(A_{ij}(t) = A_{ij}(t-1) B^t + (1-A_{ij}(t-1)) C^t\), where \(B^t \overset{i.i.d.}{\sim} \text{Bernoulli}(p)\) and \(C^t \overset{i.i.d.}{\sim} \text{Bernoulli}(q)\). This is equivalent to \(\mathbb{P}(A_{ij}(t) = 1 \mid A_{ij}(t-1) = 0) = q\) and \(\mathbb{P}(A_{ij}(t) = 0 \mid A_{ij}(t-1) = 1) = 1 - p\). This model aligns with the previous autoregressive networks without the restriction \(1 - p + q \leq 1\) but with node homogeneity. Unlike \cite{jiang2020autoregressive}, both \cite{zhao2019estimation}'s approach and our method do not require the first network to be sampled from the stationary distribution of the Markov chain. Additionally, our model allows \(p\) and \(q\) to depend on \(i\), \(j\), and \(t\), providing further flexibility.
 \end{example}

\begin{remark}
Comparing Example~\ref{ex:1} with Model 1 in \citet{padilla2019change} is akin to comparing apples and oranges, as discussed in their page 11. Moreover, our Markov chain Bernoulli network is more general than Example~\ref{ex:1},  setting our model apart from that of \citet{padilla2019change}.
\end{remark}
 
	Let $\Xi=\sup_{i,j,t}|\mathbb{P}(A_{ij}(t)=1|A_{ij}(t-1)=1)-\mathbb{P}(A_{ij}(t)=1|A_{ij}(t-1)=0)|\in [0,1]$, which is the total variation distance. Clearly, we have $ \Xi=0$ in independent networks and $  \Xi=1$ in almost surely deterministic networks.  Generally, $\Xi$ depicts the strength of dependence. In this paper, we assume that $\Xi\neq 1$, but we allow $\Xi\to 1$ as $T\to \infty$. 
	
	\subsection{Change-point detection setups}
	Let $\{A(t), A(t)\in \{0,1\}^{n(n+1)/2}\}_{t=1}^T$ be a sequence of  samples of networks with size $n$. Let  $\eta=\{\eta_1,\cdots,\eta_K\}$ be the collection of $K$ change-points  such that $1<\eta_1<\dots<\eta_K<T+1$. The multiple detection problem is to detect both the total number $K$ and the locations  $\eta_1,\cdots,\eta_{K}$. Let $\hat K$ and $\hat \eta_i$ be the estimators of $K$ and $\eta_i$. We shall consider the change-point detection for piecewise constant time varying connection probability matrices, i.e.,
	\begin{equation}
		\label{def:pc}
		\mathbb{E}A(t)\neq \mathbb{E}A(t-1) \text{ if and only if } t=\eta_k \ \text{for} \ 1\le k\le K.
	\end{equation} 
In our setting, $n, K, \eta_1,\cdots, \eta_K$ are functions of $T$. $K=0$ corresponds to the constant signal. An estimator is consistent if it satisfies the following definition. 
	
	\begin{definition}[Consistency]
		\label{newdef:consistency}
		Denote $\Delta=\min_{k=1,\cdots,K-1} \{\eta_{k+1}-\eta_{k}\}\wedge(\eta_1-1)\wedge (T+1-\eta_K)$.
		The estimator $(\hat{\eta}_1,\cdots,\hat{\eta}_{\hat{K}})$ is consistent with respect to $(\eta_1,\cdots,\eta_{K})$ iff 
		\[
		\mathbb{P}\left( \widehat{K} = K \quad \text{and} \quad \max_{k=1,\ldots,K} | \hat{\eta}_{k} -  \eta_k| \leq \epsilon \right) \to 1,
		\] 
		where 
		$
		{\epsilon}/{\Delta}\to 0 \text{ as } T\to \infty$. 
	\end{definition}

Let $\kappa = \min_{t=\eta_1,\cdots,\eta_K}||\bb{E}A(t)-\bb{E}A(t-1)||_{op}$, which  represents the signal. \cite{wang2021optimal} proposed a similar definition of signal, with the operator norm in $\kappa$ replaced by the Frobenius norm. As a consequence, they construct the statistic by taking the corresponding inner product of two independent samples (or splitting a single sample into two parts), while we directly use the operator norm (see Line 4 in Algorithm~\ref{alg:1}) of the adjacency matrix.
Using the operator norm controls the noise  better, resulting in an optimal detection bound without data splitting as done in \cite{wang2021optimal}, under the presence of dependence between edges at different times and the presence of the low-rank structure of the connection probability matrices.  

	\section{Change-point estimation}
	\label{sec:Change-point estimation}
	
	\subsection{Methodology}
	\label{newsec:K}
	
	Different from methods that require recursion or iterative optimization, we propose a new procedure based on random interval distillation to detect $\eta_k$. In the first step, we construct for each change-point $\eta_k$ a small enough interval $[u_k,v_k]$ that covers it only. In the second step, we localize  $\eta_k$ further within the interval $[u_k,v_k]$. 
	
	\begin{algorithm}[t]
		\caption{Distillation Step}
		\label{alg:1}
		\begin{algorithmic}[1]
			\Require $\{A(t)\}_{t=1}^T$, $\tau$, $M$
			\State $S=\varnothing$
			\State Take $M$ random intervals $\left\{(s_m,e_m]\right\}_{m=1}^M$ independently and uniformly
			\For{$m=1,\cdots,M$} 
	\If{$\max_{t:s_m<t<e_m}||\tilde A_{s_m,e_m}^t||_{op}>\tau$}
			\State $S=S\bigcup \left\{(s_m,e_m]\right\}$ 
			\EndIf
			\EndFor
			\State $\tilde S=S,i=1$
			\While{$|\tilde S|>0$}
			\State 
			$r_i=\mathop{\min}\{v: \exists u, (u,v]\in \tilde S\}$
			\State 
			$v^*=r_i, u^*=\mathop{\max}\{u: (u,v^*]\in \tilde S\}$
			\State $\tilde S=\tilde S\cap\{(u,v]\in \tilde S:  (u,v]\cap (u^*,v^*]\neq \varnothing\}^c$
			\State $i=i+1$
			\EndWhile
			\State $\hat K_0=i-1, \tilde S=S,i=1$
			\While{$|\tilde S|>0$}
			\State $l_{\hat K_0+1-i}=\mathop{\max}\{u: \exists v, (u,v]\in \tilde S\}$
			\State 
			$u^*=l_{\hat K_0+1-i}, v^*=\mathop{\min}\{v: (u^*,v]\in \tilde S\}$
			\State $\tilde S=\tilde S\cap\{(u,v]\in \tilde S:  (u,v]\cap (u^*,v^*]\neq \varnothing\}^c$
			\State $i=i+1$
			\EndWhile
			\State $\hat K=i-1, S^*=\left\{\left[l_{j},r_{j}\right]\right\}_{j=1}^{\hat K}$
			\State Output $\hat{K}\text{ and }S^*$
		\end{algorithmic}
	\end{algorithm}
 
  \begin{figure}[ht]
  	\centering
  	\includegraphics[width=13cm]{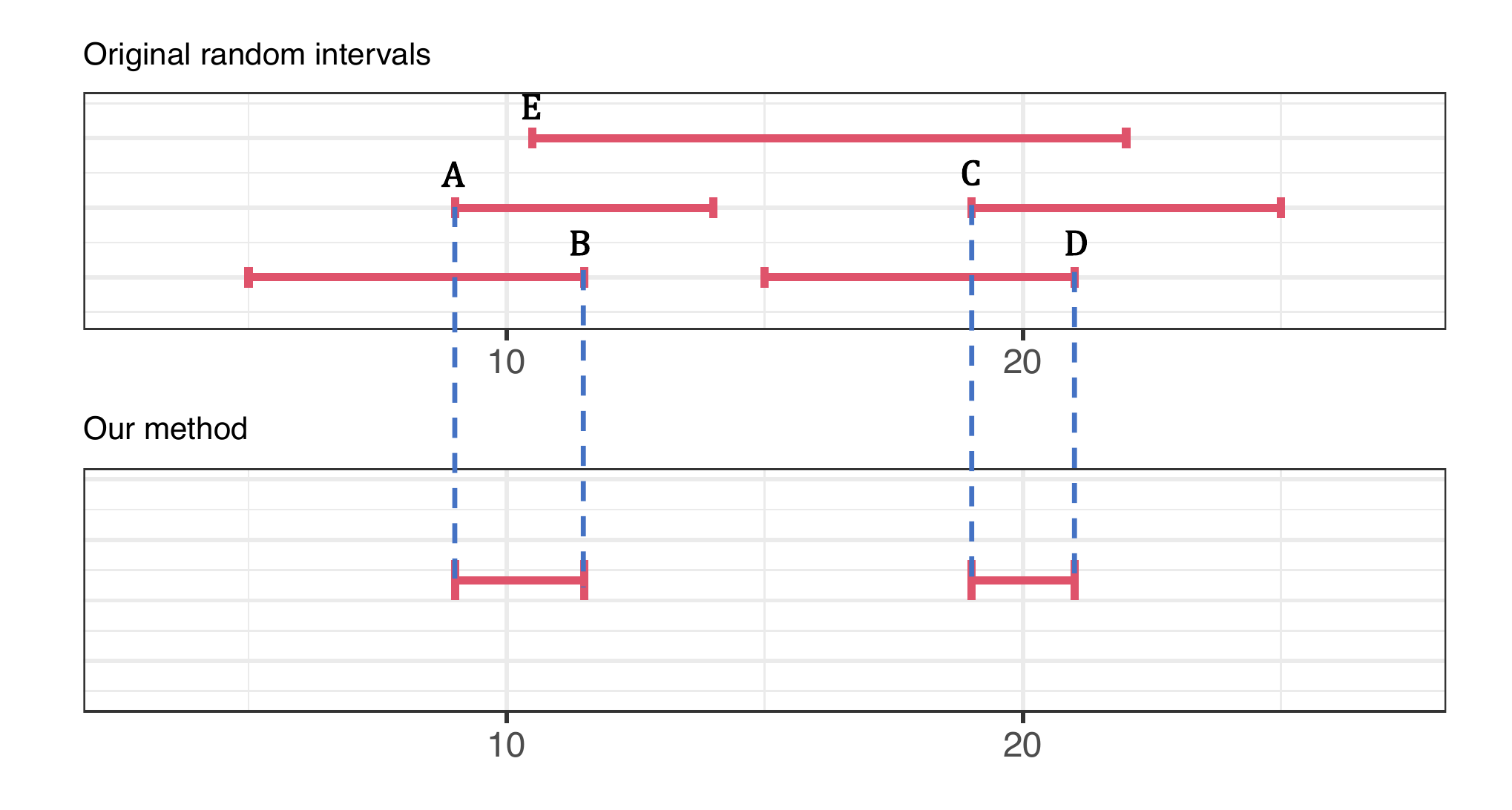}
  	\caption{Suppose that there are two change-points $\eta_1=10, \eta_2=20$. The first figure displays the random intervals (5 intervals in total, represented by horizontal lines) all with signals exceeding the threshold. Our distillation begins with finding the right endpoint $B$ (the smallest right endpoint), then the intervals corresponding to $A$ and $E$ are removed due to the overlap, and we select $D$. Similarly, for the left endpoints, we successively choose $C$ and $A$. Therefore, the final constructed intervals are $[A, B]$ and $[C, D]$. We mention that our method appropriately avoids generating $[E,B]$ which does not contain $\eta_1$.}
  	\label{fig:intro not}
  \end{figure}
 
For any triple $(s,t,e)$ where $(s,e] \subset (0,T], s<t<e$, we define the conventional  CUSUM statistics \citep{page1954continuous} as 
\begin{equation}
		\label{newdef:cusum}
		\tilde A_{s,e}^t=\sqrt{\frac{e-t}{(e-s)(t-s)}}\sum_{r=s+1}^t A(r)-\sqrt{\frac{t-s}{(e-s)(e-t)}}\sum_{r=t+1}^e A(r).
\end{equation}
Our Algorithm~\ref{alg:1} takes the data $\{A(t)\}$, threshold $\tau$ and number of random intervals $M$ as inputs, while outputs $\hat K$ and a set of intervals $S^*$.  In Line 2 in Algorithm~\ref{alg:1}, we draw each random interval by sampling two endpoints, $s,e$,  independently and  uniformly from $1,\cdots,T$. Then the interval is $[s\wedge e,s\vee e]$. In Lines 3-7 in Algorithm~\ref{alg:1}, we gather into the set $S$ the intervals $(s_m,e_m]$ in which $\max_{t:s<t<e}||\tilde A_{s,e}^t||_{op}$ is larger than $\tau$. The main purpose of this step is to exclude intervals that contain no change-point.

\begin{remark}[The necessity of the operator norm]
	\label{rem:necessity of the operator norm}
 Note that our procedure is based on the operator norm of \(\tilde{A}_{s,e}^t\) instead of the conventional Frobenius norm, since for \(\|\tilde{A}_{s,e}^t\|_F\), even if there is no change in \((s,e]\), the random noises accumulate, making \(\|\tilde{A}_{s,e}^t\|_F\) large.  One can refer to tables in Section~\ref{sec:sim} for the inferior performance   of using the Frobenius norm.  
To address this problem, \cite{wang2021optimal} split the independent samples and considered the inner product \(\langle\tilde{A}_{s,e}^t, \tilde{A}_{s,e}^{*t}\rangle\), where \(\tilde{A}_{s,e}^t\) and \(\tilde{A}_{s,e}^{*t}\) are calculated from independent subsamples. If there are no changes in \((s,e]\), the above inner product will be small since the product of random errors centers around zero. In this paper, we show that the operator norm can effectively avoid the problem of error accumulation encountered when using the Frobenius norm, by Corollary 3.2 in \cite{bandeira2016sharp} that provides the upper bound of the expectation of the operator norm of a random matrix with independent sub-Gaussian entities. Using the operator norm further results in a nearly minimax detection bound for low-rank networks, see Remark~\ref{rem:Minimax detection bound} for more details.

\end{remark}

In Lines 8-21 in Algorithm~\ref{alg:1}, we select from $S$ some endpoints and use them to form new intervals. In Lines 8-14, a greedy method is employed to find the right endpoints $\{r_i\}$ of the newly-created intervals. By reversing the time axis and applying the same method, we get the left endpoints $\{l_i\}$. It is straightforward to show that the numbers of $r_i$ and $l_i$ are equal and they both equal to the maximum sizes of subsets of $S$ with non-overlapping intervals, and the number of $r_i$ and $l_i$ is $\hat K$. It is proved in Lemma~\ref{th:alg correctness} that $l_j<r_{j}$ whenever $1\le j\le \hat K$,  which makes $\left[l_j,r_{j}\right]$  our new distilled intervals. Moreover, with a  probability tending to $1$, $\hat K=K$ and $\left[l_j,r_{j}\right]$ only covers $\eta_j$, which is guaranteed in Theorem~\ref{th:K}. The main purpose of sampling random intervals in WBS is to generate some well-positioned intervals to overcome the signal diminishing in certain intervals during the binary segmentation procedure, especially when the interval contains more than two changes. Our method (RID) inherits the benefit of sampling random intervals, and further screens the endpoints of those intervals to form the shortest possible intervals, which leads to better finite sample performance.

\begin{remark}[Comparsion with \cite{kovacs2020seeded}]
A recent popular method is Seeded Binary Segmentation (SBS) proposed by \cite{kovacs2020seeded}, which uses a set of deterministic intervals instead of random intervals. However, it is difficult to decide the threshold in practice for change-point detection of the dependent network.   See Figure~\ref{parafig} which visualizes this effect by comparing the intervals given by \cite{kovacs2020seeded} (Figure~\ref{parafig:8}) and the random intervals (Figure~\ref{parafig:1}).  The figures also demonstrate that a clustering-based method can be built on the random intervals.  

\end{remark}
\subsection{Localization optimized for low-rank networks with unknown ranks}

 To localize the change-points, a naive method is to take $\tilde \eta_k=\arg\max_{t:s_k<t<e_k}||\tilde A_{s_k,e_k}^t||_{op}$ for each $[s_k,e_k]$ in the set $S^*$ produced by Algorithm~\ref{alg:1} as the estimator of $\eta_k$.  However, this is \textit{sub-optimal} for networks with low-rank structures such as the ubiquitous communities, or RDGP (\cite{young2007random}). In Algorithm~\ref{alg:3}, we propose the sparse  USVT (SUSVT) approach for Markov chain Bernoulli networks which achieves a minimax optimal localization rate for low-rank networks, admitting dependence among adjacency matrices at different times. 

Specifically, after applying Algorithm~\ref{alg:1},  we take  data points sparsely, i.e., every $\tau_3\log(T)$ in each interval $[s_k,e_k]$ ($\tau_3$ is determined in Theorem~\ref{th:Theoretical result of local refinement} and $\tau_3\log(T)$ is an integer)  and further split the sequence into two sub-sequences. Let $t_{i,k}=s_k+1+\tau_3\log(T)(i-1), 1\le i\le 1+\lfloor \frac{e_k-s_k-1}{\tau_3\log(T)}\rfloor$. We first apply the CUSUM statistics (i.e., equation (\ref{newdef:cusum})) on the sequence $\{A(t_{2i-1,k})\}_{i\ge 1}$, and apply the USVT step. We then  calculate  the CUSUM statistics on the sequence $\{A(t_{2i,k})\}_{i\ge 1}$ and get the position $\hat\eta_k^*$ where the inner product reaches the maximum (see line 12 in the Algorithm~\ref{alg:3}).  
Finally, we select a final estimator ensembling the CUSUM values within the interval $[t_{i_k^*-2,k},t_{i_k^*+2,k}]$ where $i_k^*=\lfloor(\hat\eta_k^*-s_k-1)/(\tau_3\log(T))\rfloor$ by further getting the position where the maximum value is achieved. See line 13 in  Algorithm~\ref{alg:3}. In this way we fully utilize the local information around the change-point. We present our algorithm in Algorithm~\ref{alg:3}.  

For time complexity, Algorithm~\ref{alg:1} operates with a complexity of \( O(MTn^3 + M\log(M)) \), whereas Algorithm~\ref{alg:3} has a complexity of \( O(Tn^3) \). Consequently, the overall time complexity is \( O(MTn^3) \).

 In the literature, \citet{wang2021optimal} proposes USVT (see also Algorithm \ref{alg:usvt} in the appendix) to localize change-points in low-rank networks, assuming  independence among $\{A(t)\}$ and $\{A(t')\}$ for any $t\neq t'$.

\begin{algorithm}[htbp]
	\caption{Localization in low-rank dynamic networks with unknown ranks}
	\label{alg:3}
	\begin{algorithmic}[1]
		\Require $\{A(t)\}_{t=1}^T$, $S^*$, $\tau_2, \tau_3$
		\State $k=1,\hat K=|S^*|$. Suppose $S^*=\{[l_k,r_k]\}_{k=1}^{\hat K}$.
		\State $\hat\Delta=\min_{2\le k\le \hat K}\left(\frac{l_k+r_k}{2}-\frac{l_{k-1}+r_{k-1}}{2}\right)\wedge \left(T+1-\frac{l_{\hat K}+r_{\hat K}}{2}\right) \wedge \left(\frac{l_{1}+r_{1}}{2}-1\right)$
		\While{$k\le \hat K$}
		\State $s_k=\lfloor l_k-\hat \Delta/16\rfloor, e_k=\lfloor r_k+\hat \Delta/16\rfloor, v_k=\lfloor(l_k+r_k)/2\rfloor$
		\If{$\lfloor \frac{e_k-s_k-1}{\tau_3\log(T)}\rfloor$ is an odd number}
		\State Redefine $e_k=e_k+\tau_3\log(T)$
		\EndIf
		\State $\tilde\Delta_k=\sqrt{
			\left(\lfloor \frac{v_k-s_k-1}{2\tau_3\log(T)}\rfloor  +1\right)
			\left(\lfloor \frac{e_k-s_k-1}{2\tau_3\log(T)}\rfloor-
			\lfloor \frac{v_k-s_k-1}{2\tau_3\log(T)}\rfloor\right)/{
				\left(\lfloor \frac{e_k-s_k-1}{2\tau_3\log(T)}\rfloor+1\right)}}$
		\State \begin{align*}
			\tilde Y_{s_k,e_k}^{v_k}&=
    {\tilde\Delta_k}\sum_{i=0}^{\lfloor \frac{v_k-s_k-1}{2\tau_3\log(T)}\rfloor}A(s_k+1+2\tau_3\log(T)i)/\left(\lfloor \frac{v_k-s_k-1}{2\tau_3\log(T)}\rfloor  +1\right)
			\\&-{\tilde\Delta_k}\sum_{i=1+\lfloor \frac{v_k-s_k-1}{2\tau_3\log(T)}\rfloor}^{\lfloor \frac{e_k-s_k-1}{2\tau_3\log(T)}\rfloor}A(s_k+1+2\tau_3\log(T)i)/\left(\lfloor \frac{e_k-s_k-1}{2\tau_3\log(T)}\rfloor-
				\lfloor \frac{v_k-s_k-1}{2\tau_3\log(T)}\rfloor\right)
		\end{align*}
		\State $\hat Y_k=\text{USVT}(\tilde Y_{s_k,e_k}^{v_k},\tau_2,\tilde\Delta_k)$
		\State Let \begin{align*}
			\tilde Z_{s_k,e_k}^{t}&=
			\sqrt{\frac{
					\lfloor \frac{e_k-s_k-1}{2\tau_3\log(T)}-\frac{1}{2}\rfloor-
					\lfloor \frac{t-s_k-1}{2\tau_3\log(T)}-\frac{1}{2}\rfloor
				}{\left(
					1+\lfloor \frac{e_k-s_k-1}{2\tau_3\log(T)}-\frac{1}{2}\rfloor
					\right)
					\left(
					1+\lfloor \frac{t-s_k-1}{2\tau_3\log(T)}-\frac{1}{2}\rfloor
					\right)}} \sum_{i=0}^{\lfloor \frac{t-s_k-1}{2\tau_3\log(T)}-\frac{1}{2}\rfloor}A(s_k+1+\tau_3\log(T)(2i+1))\\
			&-
			\sqrt{\frac{
					1+\lfloor \frac{t-s_k-1}{2\tau_3\log(T)}-\frac{1}{2}\rfloor
				}{\left(
					1+\lfloor \frac{e_k-s_k-1}{2\tau_3\log(T)}-\frac{1}{2}\rfloor
					\right)
					\left(
					\lfloor \frac{e_k-s_k-1}{2\tau_3\log(T)}-\frac{1}{2}\rfloor-
					\lfloor \frac{t-s_k-1}{2\tau_3\log(T)}-\frac{1}{2}\rfloor
					\right)}}\\
			&\sum_{i=1+\lfloor \frac{t-s_k-1}{2\tau_3\log(T)}-\frac{1}{2}\rfloor}^{\lfloor \frac{e_k-s_k-1}{2\tau_3\log(T)}-\frac{1}{2}\rfloor}A(s_k+1+\tau_3\log(T)(2i+1))
		\end{align*}
		\State $\hat\eta_k^*=\arg\max_{s_k+(e_k-s_k)/100<t\le e_k-(e_k-s_k)/100}<\tilde Z_{s_k,e_k}^{t},\hat Y_k>$
  \State $\hat\eta_k=\arg\max_{\hat\eta_k^*-2\tau_3\log(T)<t\le \hat\eta_k^*+2\tau_3\log(T)}    
  \|\tilde A_{\hat\eta_k^*-2\tau_3\log(T),\hat\eta_k^*+2\tau_3\log(T)}^t\|_{op}
  $.
		\State $k=k+1$
		\EndWhile
		\State Output $\{\hat\eta_k\}_{k=1}^{\hat K}$
	\end{algorithmic}
\end{algorithm}

\subsection{Theoretical results}
	\label{Consistency results}
	 Throughout this section, we assume that $T\ge 3$. To save the notation, we write $\log^*(x)=\log(x)\wedge 1$. When $x=\infty$, we define $\log^*(x)=1$.
  We first propose the results of distilled intervals in Algorithm~\ref{alg:1}.
	
\begin{theorem}[Proporties of distilled intervals]
		\label{th:K}
		Assume that for some large enough constant $C>0$, (\RNum{1}) $	\kappa\sqrt{\Delta}>8C(\sqrt{\log(T)}/(1-\sqrt{\Xi})+\sqrt{n}/\sqrt{\log^*(1/\Xi)})$, (\RNum{2}) $C(\sqrt{\log(T)}/(1-\sqrt{\Xi})+\sqrt{n}/\sqrt{\log^*(1/\Xi)})<\tau < 
			\sqrt{\Delta}\kappa/4-C(\sqrt{\log(T)}/(1-\sqrt{\Xi})+\sqrt{n}/\sqrt{\log^*(1/\Xi)})$, then we have 
		\begin{equation}
			\label{eqn:P1}
			\mathbb{P}\left(\{\hat{K}=K\}, \bigcap_{j=1}^{\hat K}\left\{\left[l_{j},r_{j}\right]  \text{ covers } \eta_j \text{ and }r_{j}-l_{j}\le \frac{\Delta}{2}\right\}\right)\ge 1-T^{-1}-\frac{T}{\Delta}exp\left\{-\frac{M\Delta^2}{32T^2} \right\},
		\end{equation}
		where  $\hat{K}$ and $\{[l_j,r_j]\}_{j=1}^{\hat K}$ are yielded by Algorithm~\ref{alg:1} with input parameters $\tau$ and $M$.
		
	\end{theorem}
	Theorem \ref{th:K} shows that under suitable signal-to-noise condition (\RNum{1}),  with a high probability, our distilled intervals will consist of $K$ disjoint intervals with each length smaller than $\Delta/2$, and each interval covers one change-point. The signal-to-noise condition (\RNum{1})  is influenced by $\Xi$ which measure the  dependence of networks at two adjacent time points. Condition (\RNum{2}) provides a range for threshold  $\tau$ such that the resulted distilled intervals satisfy the desired property (\ref{eqn:P1}). Conditions (\RNum{1}) and (\RNum{2}) allow   both  fixed and divergent $n$ and $K$.  For implementation, we propose a data-adaptive approach to select a suitable $\tau$ through clustering in Section~\ref{subsec:parac}.  
		Since \( M \) appears only on the right-hand side of (\ref{eqn:P1}), a larger value of \( M \) increases the probability of our estimators performing effectively. We recommend choosing the largest possible $M$ under computational constraints.  
	 The probability   (\ref{eqn:P1}) will tend to $1$ if
	\begin{equation}
		\label{eqn:M1}
		\frac{T^2}{\Delta^2}\log\left(\frac{T}{\Delta}\right)=o(M).
	\end{equation}
	For WBS-based algorithms, (\ref{eqn:M1}) is commonly assumed, as discussed in \cite{fryzlewicz2014wild} among others.  
 
	\begin{remark}[Minimax detection bound for low-rank networks]
		\label{rem:Minimax detection bound}
		Let 
		\begin{align}
			\label{def:rank}
			r=\max_{1\le k\le K}\text{rank}(\mathbb{E}A(\eta_k)-\mathbb{E}A(\eta_k-1)).
		\end{align}
		 When $r$ is bounded and $\Xi$ is bounded away from $1$, our condition (\RNum{1})   can be implied by 
		\begin{equation}
			\label{minimax:2}
			\kappa_2\sqrt{\Delta}/n\ge C_1\sqrt{\log(T)}
		\end{equation}
		for some sufficiently large $C_1$, where 
  $\kappa_2 = \min_{t=\eta_1,\cdots,\eta_K}||\bb{E}A(t)-\bb{E}A(t-1)||_{F}$. 
  By  Lemma 1 and its proof in \cite{wang2021optimal}, if $T$ is sufficiently large and the joint distribution $F_A$ of $(A(1),\cdots,A(T))$ satisfies 
		$\kappa_2 \sqrt{\Delta}/n\lesssim\sqrt{\log(T)}
		$, we have 
		$\inf_{\hat \eta}\sup_{F_A}\mathbb{E}(H(\eta,\hat\eta))\ge {\Delta}/{2}$ where $H$ is the Hausdorff distance (see (\ref{def:normalized Hausdorff distance}) in Section~\ref{sec:sim} for the exact definition). Note that by Theorem~\ref{th:K}, under condition (\ref{minimax:2}), the native estimator, which selects an arbitrary value within each of the distilled intervals, satisfies $\bb{E}H(\eta,\hat\eta) \leq \Delta/4$. Therefore, our signal-to-noise ratio condition is minimax optimal in the sense that we have reached the weakest possible detection condition. Note that our method does not require prior knowledge of the minimal spacing $\Delta$. 
	\end{remark}
 
	\begin{remark}[Comparsion with \cite{wang2021optimal}]
\cite{wang2021optimal} considered change-points analysis for independent Bernoulli networks sequence. Their method splits the sequence into two independent sub-sequences and uses the inner product of the  sub-sequences,  achieving (nearly) minimax optimality for the detection bound if the minimal spacing $\Delta$ is known.  Our method achieves the nearly optimal detection bound allowing dependence among low-rank networks, and does not require prior knowledge of $\Delta$ or data splitting. 
\end{remark}

We now investigate the localization rate when the networks possess certain low-rank structures. Recall that $\mathbb{E}A(t)$ is piecewise constant and $r=\max_{1\le k\le K}\text{rank}(\mathbb{E}A(\eta_k)-\mathbb{E}A(\eta_k-1))$. We mention that for networks with community structure, $r$ is typically small because $r\le 2\max_t \text{rank}(\mathbb{E}A(t))$. 

\begin{assumption}
	\label{assumption:local refinement}
	Assume that (a) There exists a constant $c_7>1$ such that $n\le T^{c_7-\frac{1}{2}}$. Moreover, ${T^{-c_7}}\le \mathbb{E}A_{ij}(t)\le 1-T^{-c_7}$ for all $1\le i,j\le n, 1\le t\le  T$. 
 (b) ${\log(T)}=o(\Delta\log^*(1/\Xi))$. 
	(c) $\log^{3/2}(T)\sqrt{nr}=o(\kappa_2\sqrt{\Delta\log^*(1/\Xi)})$. 
\end{assumption}

 Assumption~\ref{assumption:local refinement} (a) restricts the divergence rate of the network size $n$ due to the complex dependencies among networks at different time points. When the networks are independent,  following the proof, (a) is not needed. (b) ensures that  $\log(T)/\log^*(1/\Xi)$ below is small enough, which is necessary for the localization step. (c) is the signal-to-noise ratio condition. Comparing (c) with condition (\RNum{1}) in Theorem~\ref{th:K}, the main enhanced term is a $\sqrt{r}$ term. When $r=O(1)$, (c) is asymptotically equivalent (up to a logarithmic factor) to condition (\RNum{1}) in Theorem~\ref{th:K} because $\kappa_2/\sqrt{r}\le \kappa\le\kappa_2$ and $1-\sqrt{\Xi}<\sqrt{\log^*(1/\Xi)}$. When $\Xi$ is bounded away from $1$,  (c) matches the Assumption 3 in \cite{wang2021optimal}.  

\begin{theorem}[Localization rate of change-point estimation]
	\label{th:Theoretical result of local refinement}
	Assume that the conditions for Theorem~\ref{th:K} and Assumption~\ref{assumption:local refinement}  hold.  Apply Algorithm~\ref{alg:1} to get $\hat K$ and $S^*$, and then apply Algorithm~\ref{alg:3} with input parameters $S^*, \tau_2,\tau_3$ to get $\{\hat\eta_k\}_{k=1}^{\hat K}$ where
	$\tau_2=C_7(\sqrt{n}+\sqrt{\log(T)})$, $ \tau_3=\lfloor{c_8\log(T)}/\log^*(1/\Xi)\rfloor/\log(T)$ for  arbitrary positive constants $C_7>8,c_8\ge c_7$. Then we have $\bb{P}\left(\hat{K}=K,\max_{k=1,\cdots,K} |\hat\eta_{k}-\eta_{k}|\le \varepsilon \right)\to 1$ as $T\to \infty$, where
	\begin{equation}
		\label{eqn:minimax network refinement}
		0<{\varepsilon}\le C_{4}^*\frac{c_8\log^3(T)}{\log^*(1/\Xi)\kappa_2^2},
	\end{equation}
	and $C_4^*$ is an absolute constant.
\end{theorem}
In practice, we recommend choosing $c_8=1.5$ for network with moderate size. 

\begin{remark}
	By Assumption~\ref{assumption:local refinement}, we have ${\varepsilon}/{\Delta} =o(1)$. Hence, the estimator in Theorem~\ref{th:Theoretical result of local refinement} is consistent. Furthermore, when $\Xi$ is bounded away from 1,  (\ref{eqn:minimax network refinement}) matches with \cite{wang2021optimal}, which is also  nearly minimax optimal up to a logarithmic factor,  following Lemma 2 in \cite{wang2021optimal}.
\end{remark}

It is often assumed that there is no self-loop, i.e., $A_{ii}(t)=0$ for all $i,t$, in which case the rank of $\mathbb{E}A(t)$ is not low. To fix this problem, assume that there exists a $B(t)$ such that $\bb{E}A(t)=B(t)-diag(B(t))$. Let $\tilde r=\max_t rank(B(t))$. We modify Assumption~\ref{assumption:local refinement} to the following assumption with further restrictions on $B(t)$.
\begin{assumption}
	\label{assumption:local refinement 1}
	Assume that (a) There exists a constant $c_7>1$ such that $n\le T^{c_7-\frac{1}{2}}$. Moreover, ${T^{-c_7}}\le B_{ij}(t)\le 1-T^{-c_7}$ for all $1\le i,j\le n, 1\le t\le  T$. (b) ${\log(T)}=o(\Delta\log^*(1/\Xi))$.
	(c) $\log^{3/2}(T)\sqrt{n\tilde r}=o(\kappa_2\sqrt{\Delta\log^*(1/\Xi)})$.  (d) $\|B(t))\|_F \ge C \|diag (B(t))\|_F$ for all $t$ and a large enough absolute constant $C>0$.
\end{assumption}

Notice that Assumption~\ref{assumption:local refinement 1}(d) is also considered in 
\cite{wang2021optimal}. 
\begin{theorem}
	\label{th:network refinement sbm}
	Under Assumption~\ref{assumption:local refinement 1}, Theorem~\ref{th:Theoretical result of local refinement} still holds.
\end{theorem}

\section{A clustering-based data-driven approach for $\tau$}
	\label{subsec:parac}
The selection of threshold ($\tau$ in our algorithm) is crucial in change-point detection algorithms, including the MOSUM-based method, binary segmentation-based method, and the interval distillation. The popular sSIC criterion for choosing threshold, advocated by \cite{fryzlewicz2014wild,baranowski2019narrowest}, is designed for independent data and bounded $K$, and often underestimates the threshold when correlation exists. 
	
Notice that the range of $\tau$ in condition (\RNum{2}) of Theorem~\ref{th:K} involves the key quantities $\Xi$ and $\kappa$, which are difficult to estimate especially when $\{A_{t}\}$ are heterogeneous. In this section, we are devoted to developing a change-point threshold selection method that is valid for Markov chain Bernoulli network without the need of calculating $\Xi$ and $\kappa$. 
For this purpose, define $\tau_{ref}=\max_{j=1,\cdots,T-h} \max_{j<t<j+h}||\tilde{A}_{j,j+h}^t||_{op} e_T \text{ where } h=\lfloor 3\log(T)\rfloor$ and $e_T$  is a divergent sequence that grows slowly (for example, we implement with $e_T=\log\log(T)/2$). The magnitude of order of $\tau_{ref}$ will asymptotically in the range of condition (\RNum{2}), exceeding its lower bound by a factor of $e_T$, which guaruantee the good performance of Algorithm~\ref{alg:1} when there are no change-points.  For illustrations, note that $h$ is small such that when $(j,j+h]$ contains no change-points, $\max_j \max_{j<t<j+h}||\tilde{A}_{j,j+h}^t||_{op}$ is an approximation of the left-hand side of condition (\RNum{2}).  On the contrary, $ \max_{j<t<j+h}||\tilde{A}_{j,j+h}^t||_{op} e_T$ cannot exceed the right-hand side of condition (\RNum{2}) even if it contains change-points because $h$ is small. In fact, one can use $\tau_{ref}$ as a rule of thumb threshold.  
	
	For refinements, we propose a clustering-based data-driven threshold utilizing $\tau_{ref}$. We first set the interval $[0.1\tau_{ref},  10\tau_{ref}]$ and employ a clustering-based machine learning approach to search a candidate threshold inside the interval. To see the equivelance between choosing a threshold and clustering, notice that heuristically if $\max_{s_m<t<e_m}||\tilde{A}_{s_m,e_m}^t||_{op}$ is large, it is highly probable that $(s_m,e_m)$ contains change-points (one can refer to the blue points in Figure~\ref{parafig:1}). Conversely, a small $\max_{s_m<t<e_m}||\tilde{A}_{s_m,e_m}^t||_{op}$ indicates that with a high probability  $(s_m,e_m)$ contains no change-point (one can refer to the red points in Figure~\ref{parafig:1}). Figure \ref{parafig} illustrates this phenomenon graphically in more detail. We assume the following assumption for Theorem \ref{the:gmm} which proves the above observation. 
	
	\begin{assumption}
		\label{ass:gmm} Recall that $M$ has a theoretical range (\ref{eqn:M1}).
		Assume that $\sqrt{\log(T)}/(1-\sqrt{\Xi})+\sqrt{n}/\sqrt{\log^*(1/\Xi)}=o(
		\kappa \sqrt{T/(KM)})$. 
	\end{assumption}
	
	Assumption~\ref{ass:gmm} is a technical condition, which is a strengthened version of condition (\RNum{1}). When $\Delta=\Omega(T)$, this assumption reduces to condition (\RNum{2}) because $M$ is allowed to diverge at an arbitrarily slow rate in this case. Note that  \cite{fryzlewicz2014wild}  proposed the sSIC and proved its validity when $K$ is bounded and the data is an independent sequence. In contrast, our method is applicable to scenarios where $K$ diverges. For example, when $K=T^\epsilon$ with a small positive $\epsilon<1/3$ and the change-points are approximately equally spaced, via \eqref{eqn:M1} we can choose $M=\Omega(T^{2\epsilon}\log^2 T)$,
	then the  term  ${\sqrt {KM/T}}$ is $T^{1.5\epsilon-0.5}\log T\rightarrow 0$. 
	Note that a typical magnitude of $\kappa$ in low-rank network is of the order $n$.
	
	\begin{theorem}
		\label{the:gmm} 
  Let $\mathfrak{I}=C_3(\sqrt{\log(T)}/(1-\sqrt{\Xi})+\sqrt{n}/\sqrt{\log^*(1/\Xi)})$. 
		The event 
		\begin{equation}
			\mathcal{A}:=\left\{\forall s<t<e,  \left|
			||\tilde A_{s,e}^t||_{op}-||\bb{E}\tilde A_{s,e}^t||_{op} 
			\right|\le \mathfrak{I} \right\}
		\end{equation} 
	 satisfies $\bb P(\mathcal{A})\ge 1-T^{-1}$ for some sufficiently large absolute constant $C_3>0$. Let $b$ be an arbitrary real number such that $b=o(\kappa\sqrt{T}/\sqrt{KM}), \mathfrak{I}=o(b)$. Under Assumption~\ref{ass:gmm}, we have with a probability tending to $1$, 
		\begin{enumerate}
			\item 
			$\bb E\left(\max_{s_m<t<e_m}||\tilde{A}_{s_m,e_m}^t||_{op}
			\bigg|\{|(s_m,e_m]\cap \{\eta_1\cdots,\eta_K\}|=0\},\mathcal{A}\right)<\mathfrak{I} $.
			\item 
			$\bb E\left(\max_{s_m<t<e_m}||\tilde{A}_{s_m,e_m}^t||_{op}
			\bigg|\{|(s_m,e_m]\cap \{\eta_1\cdots,\eta_K\}|>0\},\mathcal{A}\right)>b/2$.
			\item
			$\bb P\left(\bigcup_{m=1}^M\left\{\mathfrak{I} 
			<\max_{s_m<t<e_m}||\tilde{A}_{s_m,e_m}^t||_{op}< b/2\right\}\bigg|\mathcal{A}\right)\to 0.$
		\end{enumerate}
	\end{theorem}
	
	Theorem~\ref{the:gmm}(1)(2) suggest that $\max_{s_m<t<e_m}||\tilde{A}_{s_m,e_m}^t||_{op}$ can be grouped into two categories based on whether $(s_m,e_m]$ contain change-points, with the categories exhibiting well-separated conditional  expectations. Theorem~\ref{the:gmm}(3) shows that with a  probability tending to $1$, all the $\max_{s_m<t<e_m}||\tilde{A}_{s_m,e_m}^t||_{op}$ fall into $[0,\mathfrak{I} ]\cup [b/2,\infty)$. Therefore, there is a ``gap'' $(\mathfrak{I} ,b/2)$, which is presented graphically in Figure~\ref{parafig:1}. Moreover,  all values in the interval ($\mathfrak{I},b/2$) are qualified thresholds (note that we can take a large enough $C_3$). Therefore, Theorem~\ref{the:gmm} indicates that we could find  a decision boundary for clustering $\max_{s_m<t<e_m}||\tilde{A}_{s_m,e_m}^t||_{op}$ according to whether the random intervals contain change-points as a threshold.
	
	To extract a valid threshold,  we employ a clustering-based approach \citep{rodriguez2014clustering}, which recognizes the cluster centers by finding density peaks, and is flexible and suitable for  more general settings which allow heteroscedasticity and non-gaussianity.  If one of the clustering boundaries falls into the candidate range $[0.1\tau_{ref},  10\tau_{ref}]$, we set it as the threshold, otherwise we use $\tau_{ref}$ as the threshold. From our numerical studies, the clustering-based method is effective for a wide range of practical data including our Markov chain Bernoulli networks. The approach consists of  three steps.
	\begin{enumerate}[label=(\alph*)]
		\item For each point $\max_{s_m<t<e_m}||\tilde{A}_{s_m,e_m}^t||_{op}$, calculate its density $\rho_m$ through the kernel density estimation \[
		\rho_m=\sum_{i=1}^M K\left(({\max_{s_m<t<e_m}||\tilde{A}_{s_m,e_m}^t||_{op}
			-\max_{s_i<t<e_i}||\tilde{A}_{s_i,e_i}^t||_{op}})/{h}\right)/(Mh)
		\] 
		where $K(\cdot)$ is the kernel function and $h$ is a bandwidth. 
		
		\item Define \begin{equation}
			\label{paraeqn:deltam}
			\delta_m:=\min_{j:\rho_j>\rho_m} |
			\max_{s_j<t<e_j}||\tilde{A}_{s_j,e_j}^t||_{op}-
			\max_{s_m<t<e_m}||\tilde{A}_{s_m,e_m}^t||_{op}|.
		\end{equation}
		If $\{j:\rho_j>\rho_m\}=\varnothing$, define $\delta_m=\max_j |
		\max_{s_j<t<e_j}||\tilde{A}_{s_j,e_j}^t||_{op}-
		\max_{s_m<t<e_m}||\tilde{A}_{s_m,e_m}^t||_{op}|$. Construct the decision graph (see \cite{rodriguez2014clustering}) $\{(\rho_m,\delta_m)\}_{m=1}^M$.  Points with high $\delta_m$ and high $\rho_m$ are regarded as cluster centers. 
		\item Assigning the rest points to the same cluster as its nearest neighbor of higher density. For detailed allocation, we refer to \cite{rodriguez2014clustering}. 
	\end{enumerate}
	
	We display a visualization in Figure \ref{parafig}. The simulation studies demonstrate that our method for choosing thresholds works reasonably well.

 \begin{remark}
     Our results extend naturally to weighted networks with low-rank structures. For instance, the weighted RDPG \citep{larroca2021change} includes weights generated from an additional distribution, such as the Poisson distribution. For  weighted networks with bounded weights, we can derive analogous properties with minor adjustments following the proofs of Theorem~\ref{th:K} and Theorem~\ref{th:Theoretical result of local refinement}. Specifically, we can modify the definition of $\Xi$ to be the general total variation distance (see Section 4.1 in \cite{levin2017markov}) for a Markov chain. Similarly, the selection of threshold $\tau$ can be adapted by using weighted adjacency matrices instead of binary ones. We leave the specific details to future work.
 \end{remark}
 
	\begin{figure}[htbp]
		\centering 
		\tiny
		\caption{An entire threshold selection process. The original points are sampled from the Scenario 3 described in Section~\ref{sec:sim} with $(n,\kappa,T,K)=(150,5.30,160,3)$.}\subfigure[Use random intervals to calculate $\max_{s_m<t<e_m}||\tilde{A}_{s_m,e_m}^t||_{op}$ and the plot shows that the vertical coordinates of these points tend to cluster into two groups.]{
			\label{parafig:1}
			\includegraphics[width=4.5cm,height=4.5cm]{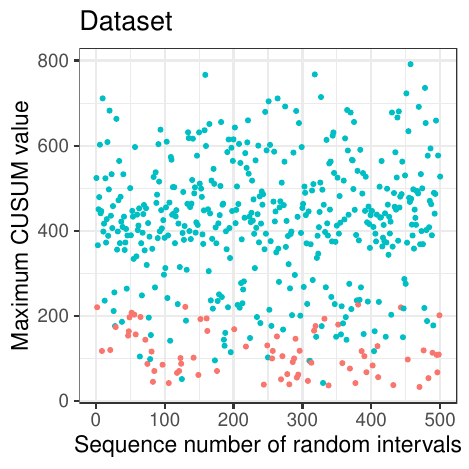}
		} 
  \hspace{0.2cm}
		\subfigure[Density of each point, evaluated by the gaussian kernel.]{
			\label{parafig:2}
			\includegraphics[width=4.5cm,height=4.5cm]{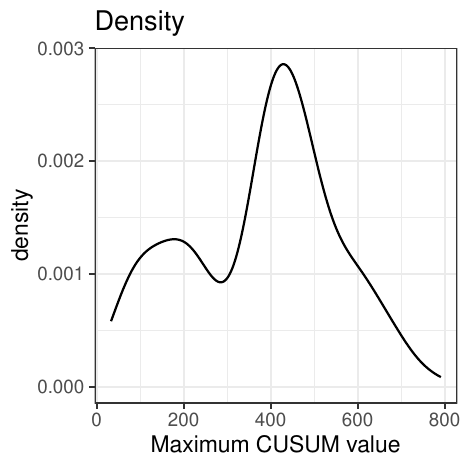}
		} 
		 \hspace{0.2cm}
		\subfigure[Calculate $\delta_m$ according to (\ref*{paraeqn:deltam}) and draw the decision graph described in \cite{rodriguez2014clustering}.]{
			\label{parafig:3}
			\includegraphics[width=4.5cm,height=4.5cm]{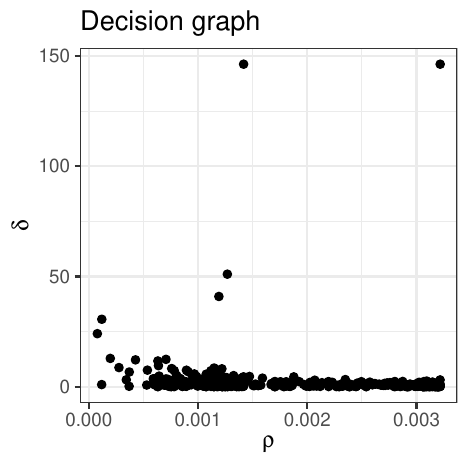}
		} 
		\subfigure[Selecting the number of clusters as 2. The boundary are shown in solid line.]{
			\label{parafig:4}
			\includegraphics[width=4.5cm,height=4.5cm]{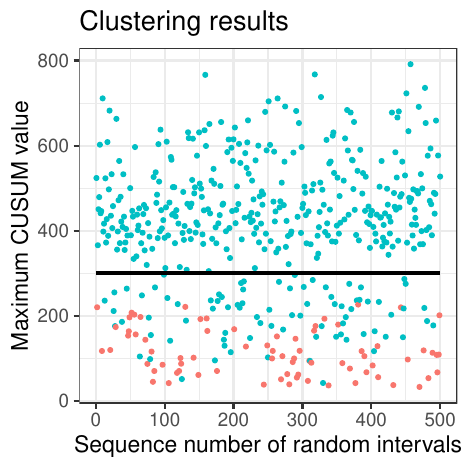}
		} 
		 \hspace{0.2cm}
		\subfigure[The simulated points $\max_{j<t<j+h}||\tilde{A}_{j,j+h}^t||_{op}$ are drawn and the rule-of-thumb threshold $\tau_{ref}$ is given.]{
			\label{parafig:5}
			\includegraphics[width=4.5cm,height=4.5cm]{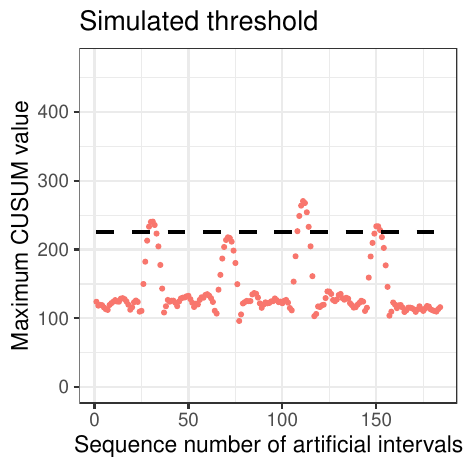}
		} 
 \hspace{0.2cm}
		\subfigure[The final threshold is presented with solid lines. The black dashed line represents $\tau_{ref}$. 
		]{
			\label{parafig:6}
			\includegraphics[width=4.5cm,height=4.5cm]{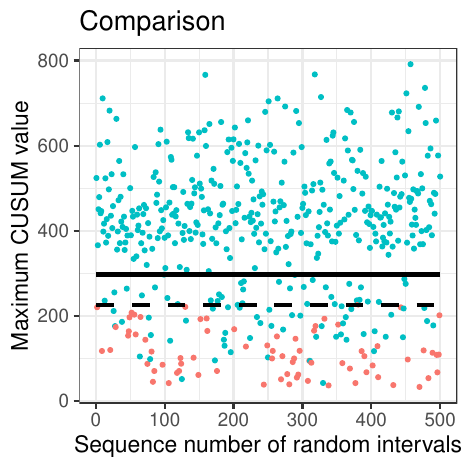}
		} 
			\subfigure[Comparison 1: If we replace the operator norm with the Frobenius norm, the density does not evidently exhibit a bimodal structure.
		]{
			\label{parafig:7}
			\includegraphics[width=5cm,height=5cm]{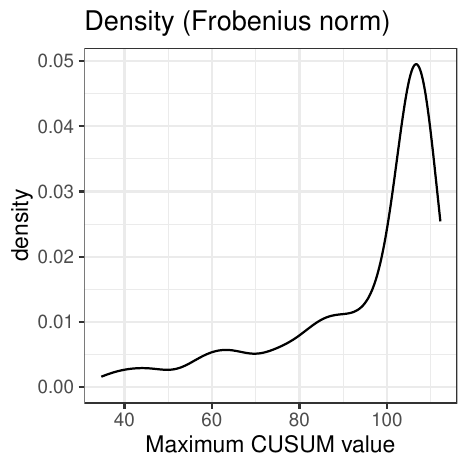}
		} 
		  \hspace{4cm}
			\subfigure[Comparison 2: When using the seeded intervals (\cite{kovacs2020seeded}), it remains an open question how to properly choose the threshold under complex dependence. 
		]{
			\label{parafig:8}
			\includegraphics[width=5cm,height=5cm]{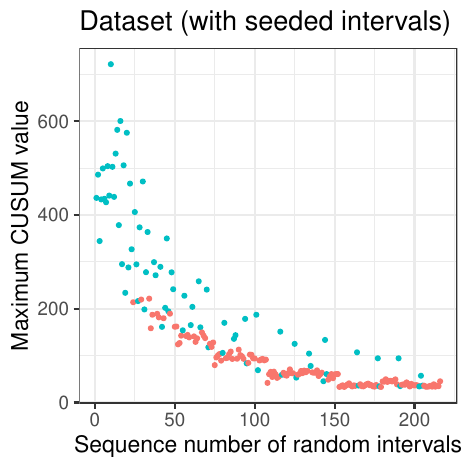}
		} 
		
		\label{parafig}
	\end{figure}
 
	\section{Simulation studies}
	\label{sec:sim}
	
	In this section, we examine the effectiveness of our method through  examples of detecting changes in block models and latent space models.  To save the notation, we let $\eta_0=1, \eta_{K+1}=T+1$.  We compare our RID (Algorithm~\ref{alg:1} and Algorithm~\ref{alg:3}) with NBS (\cite{wang2021optimal}) and NonPar-RDPG-CPD (\cite{padilla2019change}). For these methods, we use the R functions provided by the respective authors with their default parameters. It is important to note that these comparisons may not be entirely fair, as their thresholds are not designed for dependent networks or networks that extend beyond the RDPG model.  Moreover, we also compare with RID-naive where we take $\tilde \eta_k=\arg\max_{t:s_k<t<e_k}||\tilde A_{s_k,e_k}^t||_{op}$ for each $[s_k,e_k]$ in the set $S^*$ produced by Algorithm~\ref{alg:1} as the estimators, and RID-Frobenius where we replace the operator norm with the Frobenius norm in Algorithm~\ref{alg:1} and in Section~\ref{subsec:parac}.  We adopt several  measurements for evaluation,  which are commonly used in related literature.
	\begin{itemize}
		\item A list of $\hat K-K$, where $\hat K$ and $K$ are numbers of estimated and true change-points.
		\item $H(\hat\eta,\eta)/T$, the normalized Hausdorff distance defined as
		\begin{equation}
			\label{def:normalized Hausdorff distance}
			\frac{H(\hat\eta,\eta)}{T}:=\frac{1}{T}\max\left\{\max_{x\in \eta}\min_{y\in \hat\eta} |x-y|,\max_{y\in \hat\eta}\min_{x\in \eta} |x-y|\right\}.\
		\end{equation} A smaller $H(\hat  \eta, \eta)/T$ indicates a better performance.
		\item The Hausdorff distance, constrained to samples where \(\hat{K} = K\).  We denote it as $H^*(\hat\eta,\eta)$. 
		\item Averaged Rand Index \cite{rand1971objective,hubert1985comparing}, which is denoted by ARI, and ranges between 0 and 1. A higher value corresponds to a more accurate estimation. 
		\item Average time cost measured in seconds.
	\end{itemize}

For each scenario, we perform 100 repetitions. All runtime results are recorded on an Apple M1 machine with 16GB of RAM, running macOS Sonoma.

		\textbf{(Scenario 1.)} We set $n=50, K=2, (M,\Delta)=(500,50),(800,75), \eta_1=\Delta, \eta_2=3\Delta, T=4\Delta$. The connection probability for the edge $E_{ij}$ at time $t$ is given by $\Theta_{ij}(t)=0.6\sin( {f_n(i,t)} +  {f_n(j,t)})$ where \begin{align*}f_n(u,t)=\sin(u/n)(I(t<\eta_1)+I(t\ge \eta_2))+[\sin((21-u)/n)I(u\le 20)\\+\sin(u/n)I(u>20)]I(\eta_1\le t<\eta_2), u=1,2,\cdots,n.\end{align*} For each $1\le i,j\le n$, the time varying Markov transition kernels  are 
		\[
		P_{ij}(t)= \left(\begin{array}{cc}
			1-m(t)\Theta_{ij}(t) & m(t)\Theta_{ij}(t)  \\
			m(t)-m(t)\Theta_{ij}(t) & 1-m(t)+m(t)\Theta_{ij}(t)
		\end{array}\right)
		\]
		where $m(t)\sim Uniform(0.3,0.6)$ independently and identically for $t=1,\cdots,T-1$. 
		It is easy to show that within each segment, $\Theta_{ij}(t+1)=\Theta_{ij}(t)(1-m(t)+m(t)\Theta_{ij}(t))+m(t)\Theta_{ij}(t)(1-\Theta_{ij}(t))=\Theta_{ij}(t)$. By drawing random values, the sequence remains heterogeneous even in intervals without change-points. We use the data-driven $\tau$ in Section~\ref{subsec:parac}  for our method, and choose $\tau_2=0.25(\sqrt{n}+\sqrt{\log(T)}), \tau_3=3/\log(T)$.

	Table~\ref{sim:Results for Scenario 0}  reveals several insights. In the model where $\Theta(t)$ follows a  Markov process with time-dependent transition matrices, our RID performs well. Moreover, the SUSVT refinement is effective in the sense that $H^*$ for RID  outperforms that of RID-naive. Notably, there are cases where $H^*(\hat\eta,\eta)=0$, indicating that the localization results are perfectly accurate provided that the number of change-points $K$ is  estimated correctly. Both NBS and NonPar-RDPG-CPD perform unsatisfactorily.  This makes sense since NBS is proposed for independent networks and NonPar-RDPG-CPD is tailored to a specific kind of stationary networks.  Furthermore, RID-Frobenius is also not ideal, aligning with the intuition discussed in Remark~\ref{rem:necessity of the operator norm}, and supporting the necessity of employing the operator norm. Finally, our code runs significantly faster than the NBS and NonPar-RDPG-CPD methods, despite the similar theoretical complexities.
	
	 \begin{table}[htbp]
		\centering
		\tiny
		\caption{Results for heterogeneous networks with unequally spaced change-points in Scenario 1.}
		\label{sim:Results for Scenario 0}
		\begin{tabular}{@{}lcrrrrrrrccc@{}}
			\hline
			&& \multicolumn{6}{c}{$\hat{K}-K$} &&\\
			\cline{3-8}
			$(M,\Delta,T)$ &Method &
			\multicolumn{1}{c}{-2} &
			\multicolumn{1}{c}{-1} &
			\multicolumn{1}{c}{0} &
			\multicolumn{1}{c}{1} &
			\multicolumn{1}{c}{2} &
			\multicolumn{1}{c}{$\ge 3$} &
			$H^*(\hat{\eta},\eta)\cdot 10^2/T$
			& ARI &
			$H(\hat{\eta},\eta)\cdot 10^2/T$ & Time (in seconds) \\
			\hline
			
			(500,50,200) 
			& RID & 6 & 6 & 87 & 0 & 0 & 1 & 0.029 & 0.937 & 6.285 & 8.470 \\ 
			& RID-naive   & 6 & 6 & 87 & 0 & 0 & 1 & 0.034 & 0.936 & 6.320 & 8.332 \\ 
			
			& RID-Frobenius  & 31 & 0 & 7 & 10 & 4 & 48 & 0.929 & 0.661 & 36.220 & 2.248 \\
			
			& NBS   & 0 & 0 & 0 & 0 & 0 & 100 && 0.695& 23.000  &26.418  \\ 
			& NonPar-RDPG-CPD  & 57 & 28 & 14 & 1 & 0 & 0 & 21.250 & 0.504 & 57.070 & 102.874\\ 
			\hline
			
			(800,75,300) 
			& RID &   2 & 2 & 96 & 0 & 0 & 0 & 0.000 & 0.980 & 2.050 & 18.502 \\
			& RID-naive &   2 & 2 & 96 & 0 & 0 & 0 & 0.017 & 0.98 & 2.067 & 18.275 \\  
			& RID-Frobenius  & 52 & 0 & 4 & 4 & 6 & 34 & 1.583 & 0.569 & 48.367 & 5.490 \\
			
			& NBS   & 0 & 0 & 0 & 0 & 0 & 100 && 0.675& 23.520   & 117.680 \\ 
			& NonPar-RDPG-CPD  & 90 & 10 & 0 & 0 & 0 & 0   & & 0.400&71.256 & 302.634\\ 
			\hline
		\end{tabular}
	\end{table}

	\textbf{(Scenario 2.)} 
	 We now investigatie heterogeneous networks with increasing number of  change-points.  
	 We set $(n,\kappa)=(50,4.66),(150,5.30), (T,\Delta,K)=(160,40,3), (250,50,4)$ and $(360,50,5)$, and change-points are equally spaced.  The connection probabilities  (denoted as $\Theta$) of Markov chain Bernoulli networks are
	\[\Theta(t)=\begin{cases}
		\rho Z_n
		Q_1
  Z_n^\top \quad \eta_{2i}\le t<\eta_{2i+1}, i=0,\cdots,\lfloor K/2\rfloor,\\
		\rho Z_n
		Q_2Z_n^\top\quad \eta_{2i-1}\le t<\eta_{2i}, i=1,\cdots,\lceil K/2\rceil,\\
	\end{cases}\]
	where 
 \begin{align*}
     Q_1=\left(\begin{array}{ccc}
			0.4 & 1 & 0.4 \\
			1 & 0.4 & 0.4 \\
			0.4 & 0.4 & 0.4 \end{array}\right),
   Q_2=\left(\begin{array}{ccc}
			0.4 & 0.4 & 1 \\
			0.4 & 0.4 & 0.4 \\
			1 & 0.4 & 0.4 
		\end{array}\right),
  Z_n=\begin{pmatrix}
  \mathbbm{1}_{\lfloor n/3\rfloor} &0&0\\
    0&\mathbbm{1}_{\lfloor n/3\rfloor}&0\\
    0&0&\mathbbm{1}_{n-2\lfloor n/3\rfloor} 
		\end{pmatrix},
 \end{align*}
 $\rho=I(n=50)/3+I(n=150)/8$, $\mathbbm{1}_j=(1,1,\cdots,1)^\top$.  
	That is, there are three communities in the network and each community contains $\lfloor n/3\rfloor,\lfloor n/3\rfloor,n-2\lfloor n/3\rfloor$ nodes respectively. For each $1\le i,j\le n$, the time varying Markov transition kernels  are 
	\[
	P_{ij}(t)= \left(\begin{array}{cc}
		1-m(t)\Theta_{ij}(t) & m(t)\Theta_{ij}(t)  \\
		m(t)-m(t)\Theta_{ij}(t) & 1-m(t)+m(t)\Theta_{ij}(t)
	\end{array}\right)
	\]
 where $m(t)=0.1+0.8t(T-t)/T^2$.  
	We use the data-driven $\tau$ in Section~\ref{subsec:parac}  for our method, and choose $M=500, \tau_2=0.6(\sqrt{n}+\sqrt{\log(T)}), \tau_3=3/\log(T)$.   Similar to Scenario 1,  Table~\ref{sim:Results for Scenario 1} shows that RID accurately detects change-points in the presence of complex dependence, outperforming both NBS and NonPar-RDPG-CPD.

\begin{table}[htbp]
		\centering
		\tiny
		\caption{Results for heterogeneous networks in Scenario 2.}
		\label{sim:Results for Scenario 1}
		\begin{tabular}{@{}lcrrrrrrrccc@{}}
			\hline
			&& \multicolumn{6}{c}{$\hat{K}-K$} &&\\
			\cline{3-8}
			$(n,\kappa,T,K)$ &Method &
			\multicolumn{1}{c}{$\le-2$} &
			\multicolumn{1}{c}{-1} &
			\multicolumn{1}{c}{0} &
			\multicolumn{1}{c}{1} &
			\multicolumn{1}{c}{2} &
			\multicolumn{1}{c}{$\ge 3$} &
			$H^*(\hat{\eta},\eta)\cdot 10^2/T$
			  & ARI &
			$H(\hat{\eta},\eta)\cdot 10^2/T$ & Time (in seconds) \\
			\hline
			
			(50,4.66,160,3) 
			& RID &  0 & 1 & 99 & 0 & 0 & 0 & 0.000 & 0.998 & 0.137 & 6.427 \\  
			& RID-naive &   0 & 1 & 99 & 0 & 0 & 0 & 0.044 & 0.998 & 0.200 & 6.323 \\     
			& RID-Frobenius & 14 & 27 & 5 & 3 & 4 & 47 & 2.750 & 0.836 & 22.775 & 1.784\\
			& NBS  & 0 & 0 & 0 & 0 & 0 & 100  &  &0.870& 14.350& 18.319\\ 
			& NonPar-RDPG-CPD  & 65 & 24 & 9 & 2 & 0 & 0    & 0.583 & 15.902 & 22.533 & 65.480 \\ 
			\hline
			
			(50,4.66,250,4) 
			&RID &  0 & 0 & 100 & 0 & 0 & 0 & 0.000 & 1.000 & 0.000 & 9.747 \\ 
			& RID-naive  & 0 & 0 & 100 & 0 & 0 & 0 & 0.056 & 1.000 & 0.056 & 9.581 \\  
			  & RID-Frobenius & 21 & 1 & 2 & 5 & 8 & 63 & 5.004 & 0.868 & 13.624 & 2.950 \\
			 & NBS   & 0 & 0 & 0 & 0 & 0 & 100 && 0.870& 12.928 & 57.598 \\ 
			& NonPar-RDPG-CPD  & 100 & 0 & 0 & 0 & 0 & 0   & & 0.257 & 74.592 & 154.656\\ 
			\hline
			
			(50,4.66,360,5)  
			  &  RID &  0 & 0 & 100 & 0 & 0 & 0 & 0.097 & 0.999 & 0.097 & 13.720 \\  
			& RID-naive &  0 & 0 & 100 & 0 & 0 & 0 & 0.156 & 0.999 & 0.156 & 13.467 \\ 
			  & RID-Frobenius  & 21& 2 & 0 & 1 & 7 & 69 &   & 0.874 & 14.106 & 3.891 \\ 
			& NBS &   0 & 0 & 0 & 0 & 0 & 100 &   & 0.881 &12.605& 120.043 \\  
				& NonPar-RDPG-CPD  & 100 & 0 & 0 & 0 & 0 & 0   && 0.187& 81.344 & 257.272\\ 
			\hline
			
			(150,5.30,160,3) 
			& RID&   0 & 0 & 100 & 0 & 0 & 0 & 0.000 & 1.000 & 0.000 & 26.410 \\  
			 & RID-naive &   0 & 0 & 100 & 0 & 0 & 0 & 0.000 & 1.000 & 0.000 &25.686\\   
			 & RID-Frobenius & 4 & 0 & 5 & 6 & 18 & 67 & 7.500 & 0.860 & 16.200 & 13.828\\
			& NBS &   0 & 0 & 0 & 0 & 0 & 100 & &0.870  & 16.425 & 162.869  \\  
				& NonPar-RDPG-CPD  & 81 & 16 & 3 & 0 & 0 & 0 &15.833&   0.507& 51.281 & 143.522\\ 
			\hline
			
			(150,5.30,250,4) 
			 & RID   & 0 & 0 & 100 & 0 & 0 & 0 & 0.000 & 1.000 & 0.000 & 40.970 \\  
			 & RID-naive &   0 & 0 & 100 & 0 & 0 & 0 & 0.044 & 1.000 & 0.044 & 39.722 \\ 
			& RID-Frobenius& 2  & 2 & 2 & 4 & 9 & 81 & 14.000 & 0.875 & 14.188 & 20.432 \\
			& NBS &   0 & 0 & 0 & 0 & 0 & 100 &&0.869&13.248 & 460.461 \\  
			&NonPar-RDPG-CPD  & 100 & 0 & 0 & 0 & 0 & 0 &&    0.250& 75.544  & 385.760\\ 
			\hline
			
			(150,5.30,360,5) 
			  & RID & 0 & 2 & 97 & 1 & 0 & 0 & 0.052 & 0.998 & 0.461 & 59.064 \\  
			  & RID-naive   & 0 & 2 & 97 & 1 & 0 & 0 & 0.132 & 0.997 & 0.456 & 57.107 \\  
			  & RID-Frobenius& 1 & 2 & 1 & 2 & 11 & 83 & 16.667 & 0.907 & 11.303 & 29.970 \\ 
			& NBS &   0 & 0 & 0 & 0 & 0 & 100 &&0.881& 11.594& 1419.869 \\  
			&NonPar-RDPG-CPD  & 100 & 0 & 0 & 0 & 0 & 0 && 0.164& 83.333 & 905.951\\ 
			\hline
		\end{tabular}
	\end{table}
	
\begin{table}[tbp]
		\centering
		\tiny
		\caption{Results for homogeneous networks in Scenario 3.}
		\label{sim:Results for Scenario 2}
		\begin{tabular}{@{}lcrrrrrrrcc@{}}
			\hline
			&& \multicolumn{6}{c}{$\hat{K}-K$} &&\\
			\cline{3-8}
			$(n,\kappa,T,K)$ &Method &
			\multicolumn{1}{c}{$\le-2$} &
			\multicolumn{1}{c}{-1} &
			\multicolumn{1}{c}{0} &
			\multicolumn{1}{c}{1} &
			\multicolumn{1}{c}{2} &
			\multicolumn{1}{c}{$\ge 3$} &
			$H^*(\hat{\eta},\eta)\cdot 10^2/T$
			& ARI &
			$H(\hat{\eta},\eta)\cdot 10^2/T$  \\
			\hline
			
			(50,4.66,160,3) 
			& RID   & 0 & 0 & 99 & 1 & 0 & 0 & 0.000 & 0.999 & 0.119   \\  
			& RID-naive &   0 & 0 & 99 & 1 & 0 & 0 & 0.000 & 0.999 & 0.112   \\  
			& RID-Frobenius & 18 & 5 & 14 & 3 & 8 & 52 & 1.741 & 0.797 & 21.831   \\  
			& NBS &   0 & 0 & 0 & 0 & 0 & 100 &   & 0.874 & 13.863  \\  
			& NonPar-RDPG-CPD & 65 & 24 & 8 & 3 & 0 & 0 & 17.109 & 0.592 & 41.663   \\
			\hline
			
			(50,4.66,250,4) 
			& RID &  0 & 0 & 100 & 0 & 0 & 0 & 0.000 & 1.000 & 0.000  \\
			& RID-naive &  0 & 0 & 100 & 0 & 0 & 0 & 0.000 & 1.000 & 0.000 \\  
			& RID-Frobenius & 10 & 1 & 12 & 3 & 4 & 70 & 1.100 & 0.838 & 17.476   \\  
			& NBS &  0 & 0 & 0 & 0 & 0 & 100 &   & 0.871 & 13.960   \\  
			& NonPar-RDPG-CPD  & 100 & 0 & 0 & 0 & 0 & 0 &   & 0.299 & 71.396   \\ 
			\hline
			
			(50,4.66,360,5)   
			& RID &  0 & 2 & 98 & 0 & 0 & 0 & 0.000 & 0.998 & 0.197   \\  
			& RID-naive &   0 & 2 & 98 & 0 & 0 & 0 & 0.000 & 0.998 & 0.200  \\  
			& RID-Frobenius & 9  & 1 & 7 & 1 & 2 & 80 & 0.913 & 0.863 & 15.142  \\
			& NBS &   0 & 0 & 0 & 0 & 0 & 100 &   & 0.879 & 12.511 \\ 
			& NonPar-RDPG-CPD & 100  & 0 & 0 & 0 & 0 & 0 &   & 0.218 & 78.664   \\
			\hline
			
			(150,5.30,160,3) & RID   & 0 & 0 & 100 & 0 & 0 & 0 & 0.000 & 1.000 & 0.000  \\  
			& RID-naive &   0 & 0 & 100 & 0 & 0 & 0 & 0.000 & 1.000 & 0.000   \\ 
			& RID-Frobenius & 1   & 1 & 6 & 4 & 18 & 70 & 1.042 & 0.863 & 15.188  \\  
			& NBS &   0 & 0 & 0 & 0 & 0 & 100 &   & 0.887 & 16.137   \\  
			& NonPar-RDPG-CPD & 67 & 28 & 4 & 1 & 0 & 0 & 11.719 & 0.611 & 41.462 \\
			\hline
			
			(150,5.30,250,4) &    RID  &   0 & 0 & 100 & 0 & 0 & 0 & 0.000 & 1.000 & 0.000 \\  
			& RID-naive &   0 & 0 & 100 & 0 & 0 & 0 & 0.000 & 1.000 & 0.000  \\  
			& RID-Frobenius &   0 & 0 & 2 & 3 & 7 & 88 & 2.000 & 0.886 & 12.616   \\  
			& NBS &   0 & 0 & 0 & 0 & 0 & 100 &   & 0.875 & 13.184   \\  
			& NonPar-RDPG-CPD & 100  & 0 & 0 & 0 & 0 & 0 & & 0.241 & 76.140  \\
			\hline
			
			(150,5.30,360,5) & RID &   0 & 1 & 99 & 0 & 0 & 0 & 0.000 & 0.999 & 0.167  \\  
			& RID-naive &  0 & 1 & 99 & 0 & 0 & 0 & 0.000 & 0.999 & 0.114  \\ 
			& RID-Frobenius  &   1 & 2 & 5 & 5 & 5 & 82 & 2.111 & 0.907 & 10.211   \\  
			& NBS &  0 & 0 & 0 & 0 & 0 & 100 &  & 0.876 & 12.756   \\   
			& NonPar-RDPG-CPD & 100 & 0 & 0 & 0 & 0   & 0 &   & 0.164 & 83.333  \\ 
			\hline
		\end{tabular}
	\end{table}

\textbf{(Scenario 3.)} The model and parameter choices are identical to those in Scenario 2, except that $m(t)$ is set to $0.2$. In this case, the sequences between change-points are stationary. From Table~\ref{sim:Results for Scenario 2}, RID performs well with the operator norm. Similar to the findings in Scenario 2, both NBS and NonPar-RDPG-CPD show unsatisfied performance, as  they are not designed for dynamic networks with Markov structure.

	\section{Real data example}
	\label{sec:real data}
	
	In this section, we apply our method to detect change-points of connection probabilites in face-to-face interactions in a primary school \citep{gemmetto2014mitigation, stehle2011high}. The data is collected by the SocioPatterns project\footnote{\url{http://www.sociopatterns.org}} with active RFID devices, which generate a new data record every 20 seconds capturing information from the preceding 20 seconds. Specifically, on October 1st, 2009, from 8:40 to 17:18, contact data were collected for a total of 236 individuals, including 10 classes of students and one group of teachers, with a total of 60623 records. Similarly, on October 2nd, 2009, from 8:34 to 17:18, contact data were collected for 238 individuals, with a total of 65150 records. We use these data to construct minute-level dependent dynamic networks and use our method to detect change-points. Such analysis is highly beneficial for studying the contact patterns in primary schools, quantifying the opportunities for respiratory infection transmission, and maintaining the safety of primary school campuses.
	
	For the $i_{th} (i=1,2)$ day, let $A^{(i)}(t)$ denote the contact matrix for the $t_{th}$ minute, i.e.,
	\[
	A_{kl}^{(i)}(t)=\left\{
	\begin{matrix}
		1 & \text{ individuals $k$ and $l$ contacted at least once within the $t_{th}$ minute on $i_{th}$ day,}\\
		0& \text{ otherwise, }
	\end{matrix}
	\right.\]
	where $520\le t\le 1038$ for $i=1$ and $514\le t\le 1038$ for $i=2$. It is evident that for fixed $k,l,i$, the sequence $\{A_{kl}^{(i)}(t)\}$ with respect to $t$ is dependent, since children's contacts may last for a certain duration. We first analyze the data set and identify the dependencies  among networks at different time points in Section~\ref{supp sec:real data} of the appendix. Then we perform change-point detections on the expected probability of contacts (namely $\mathbb{E}A^{(i)}(t)$ for $i=1,2$), and the results are summarized in Table~\ref{tab:res1}.  We use the data-driven $\tau$ in Section~\ref{subsec:parac}  for our method, and choose $M=1000, \tau_2=0.75(\sqrt{n}+\sqrt{\log(T)}), \tau_3=1/\log(T)$. We also consider the results from NBS and NonPar-RDPG-CPD in Table~\ref{tab:res1}.  Unfortunately, NBS appears to overidentify change-points, detecting over 25 change-points out of 519 time points on October 1st (or 525 time points on October 2nd). For NonPar-RDPG-CPD, results are unavailable due to the limit of our computational resources on an Apple M1 machine with 16GB of RAM.  
	
	From a change-point perspective, the contact patterns vary between the two days, with 4 change-points on the first day and 3 change-points on the second day. We note that \cite{stehle2011high} also identified differences in some statistical features between the two days through exploratory data analysis, as depicted in their Figure 1 and 2. Some change-points, such as those around 9:50, 12:00, and 13:50, are common changes occurred on both days. 
	Moreover, this finding aligns with Figure 6 in \cite{stehle2011high}, indicating variations due to differences between lunchtime and class hours.

\begin{table}[ht]
		\centering
	\footnotesize
		\caption{Change-point locations for the probability of contacts in the primary school.}
			\label{tab:res1}
	\begin{tabular}{ccc}
		\hline
		Day                & Method          & Time of change-points      \\ \hline
		1 & RID             & 9:53, 12:01, 13:47, 15:31  \\
		2 & RID             & 9:47, 12:03, 13:57         \\
  \hline
	\end{tabular}
\end{table}

	Note that students were assigned to specific classes, and prior analyses (e.g., \cite{stehle2011high}) indicated that the class of children can serve as communities within the network. Therefore, we compute the average number of contacts per minute for within-class and between-class contacts during different time intervals, as shown in Table \ref{real data:2}. An ``inversion" is observed between around 12:00 and 13:50 for both days, where the numbers of between-class contacts significantly exceed within-class contacts during this period. Additionally, at the period from 9:47 to 12:00, the number of within-class contacts reaches its peak.
	\begin{table}[ht]
		\centering
		\caption{Average number of contacts per minute during different time periods on the two days.}
		\label{real data:2}
		\begin{tabular}{ccccccc}
			\hline
			\multicolumn{3}{c}{Day 1}                                                                     &  & \multicolumn{3}{c}{Day2}                                                                      \\ \hline
			\multirow{2}{*}{Time period} & \multirow{2}{*}{within-class} & \multirow{2}{*}{between-class} &  & \multirow{2}{*}{Time period} & \multirow{2}{*}{within-class} & \multirow{2}{*}{between-class} \\
			&                               &                                &  &                              &                               &                                \\ \cline{1-3} \cline{5-7} 
			8:40-9:53                 & 50.17                         & 10.79                       &  
			& 8:34-9:47                   & 59.58                         & 14.08                         \\
			9:48-12:01               & 84.68                      & 21.97                          &  
			& 9:48-12:03                & 91.68                     & 15.81                      \\
			12:02-13:47                 &32.94                      & 61.40                   & 
			 & 12:04-13:57                & 38.98                        & 65.40                          \\
			13:48-15:31                  & 49.76                       & 14.81                         &  
			& 13:58-17:18                 & 63.40                       & 12.68                        \\
			15:32-17:18                & 68.59                         & 12.49                      &  &                  &                      &                      \\ \hline
		\end{tabular}
	\end{table}
	
	To visually examine the contact patterns during different time periods, we draw aggregated networks of interactions in Figure~\ref{fig:realdata1} in the main article and Figure~\ref{fig:realdata2} in the appendix. Specifically, we draw an edge on the graph if two nodes were recorded by the RFID device to have contact at least three times (i.e., totaling at least one minute) within the respective time period. Nodes are placed in an orderly manner  according to the children's classes to clearly show the difference between within-class and between-class contacts. The figures illustrate noticeable variations in contact patterns among individuals during different  periods.
	\begin{figure}[htbp]
		\centering 
		\subfigure{
			\includegraphics[width=6cm,height=6cm]{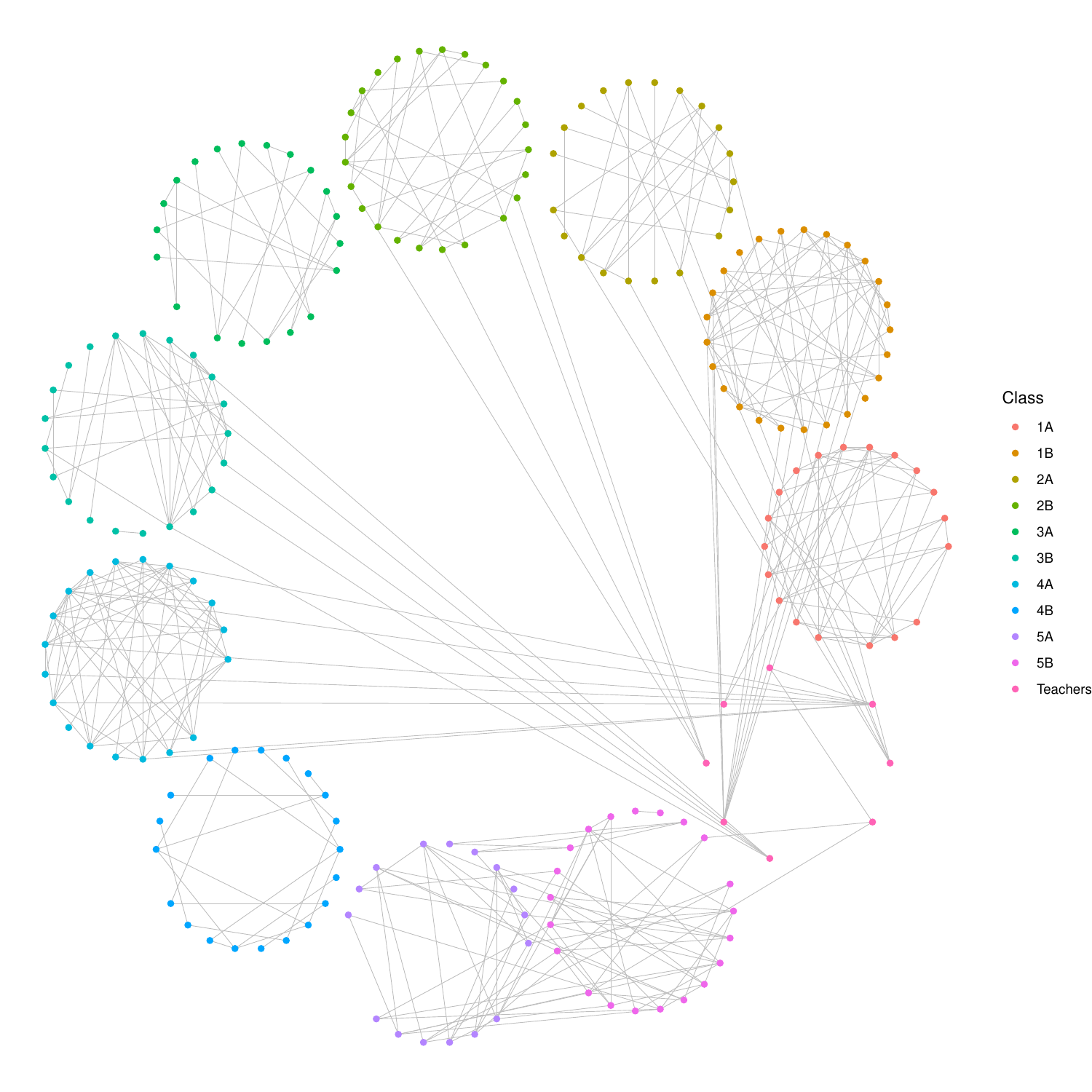}
		}\subfigure{
			\includegraphics[width=6cm,height=6cm]{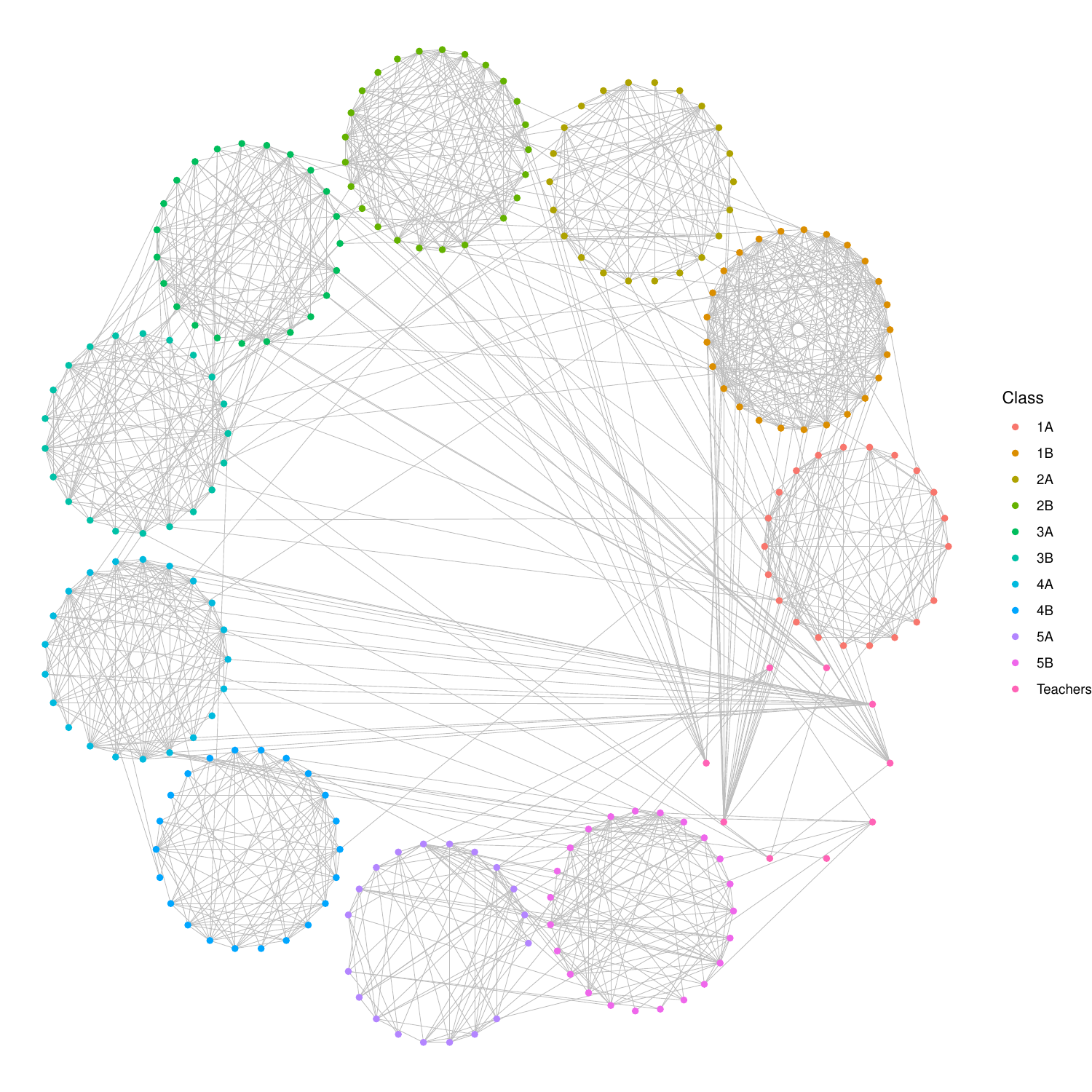}
		} 	\subfigure{
			\includegraphics[width=6cm,height=6cm]{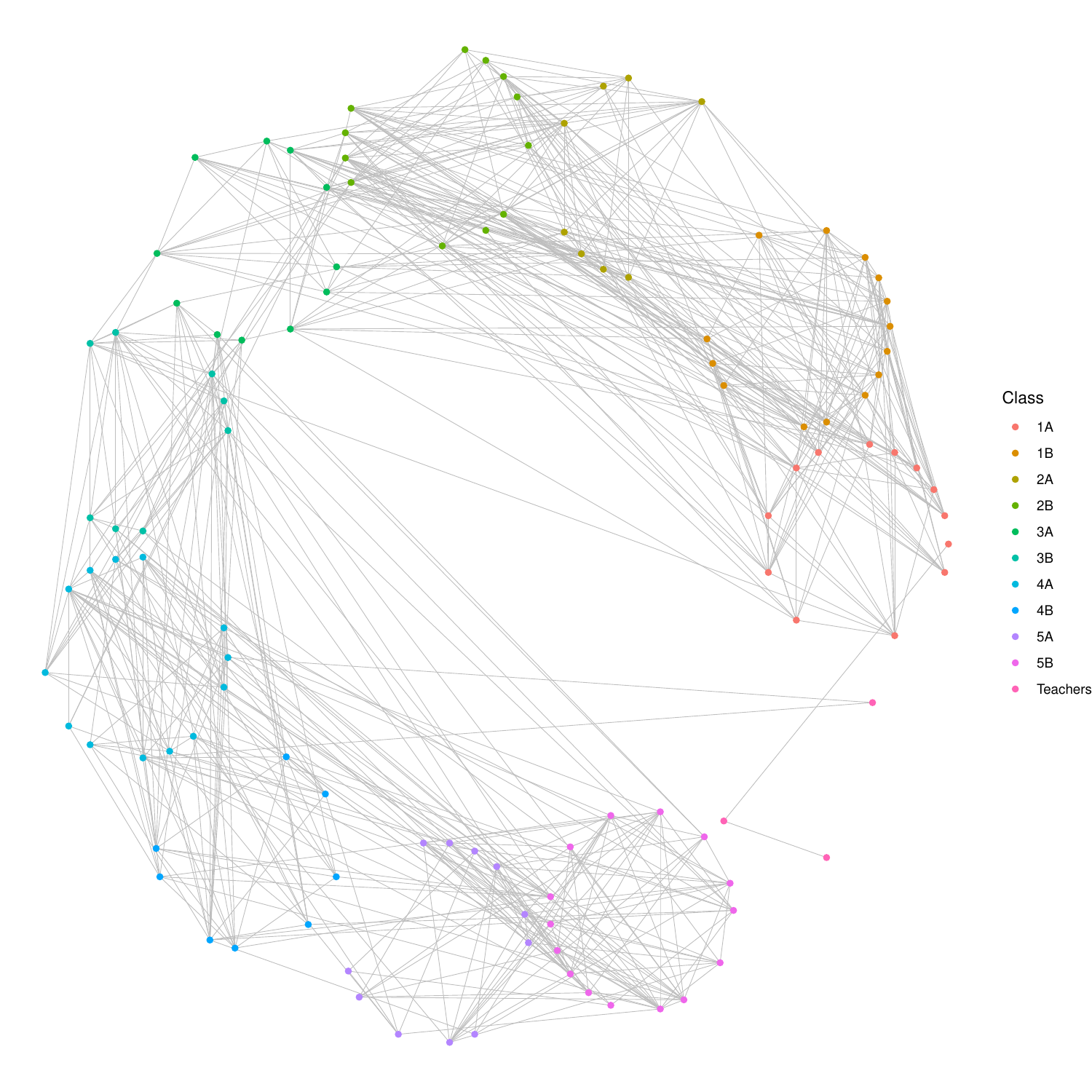}
		} 	\subfigure{
			\includegraphics[width=6cm,height=6cm]{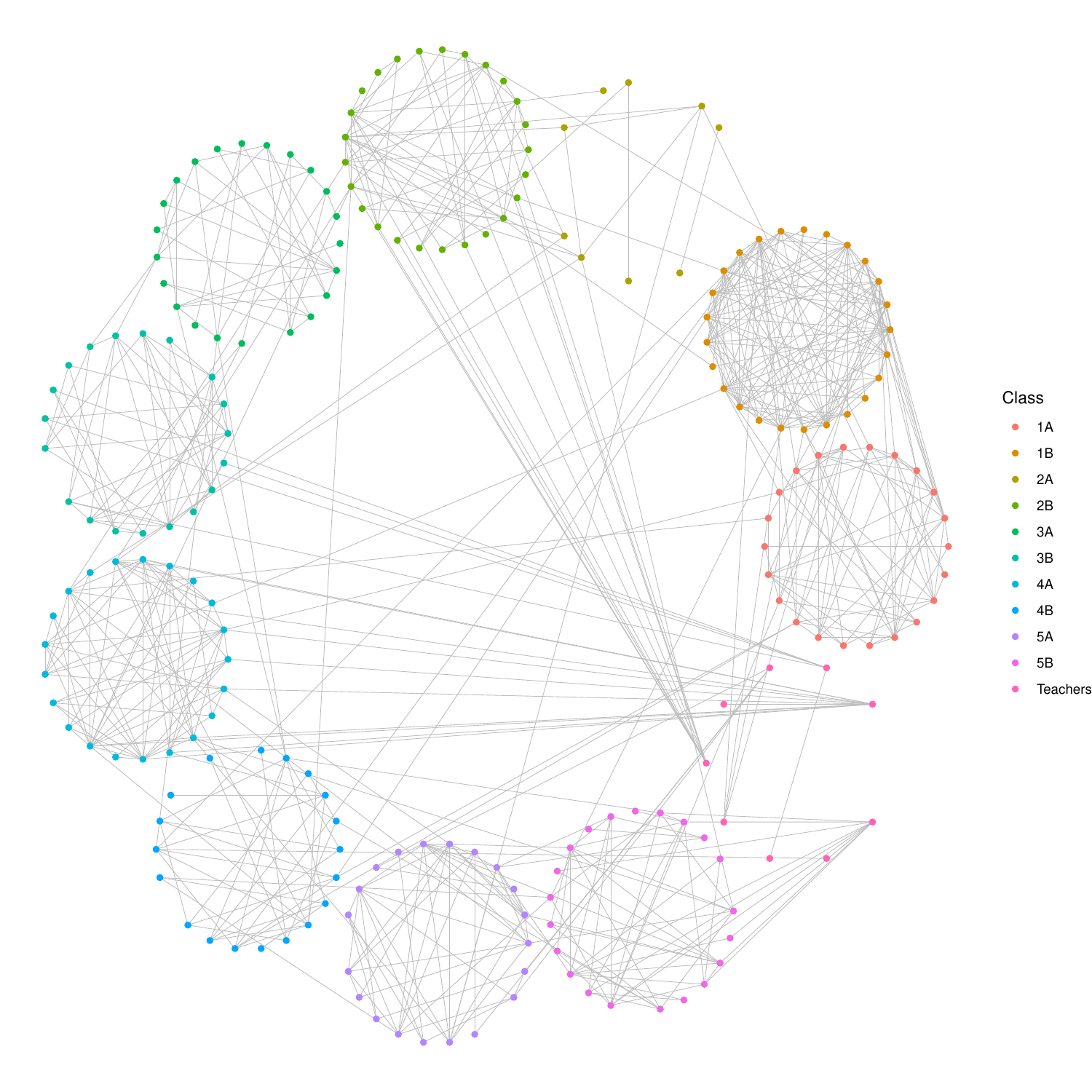}
		} 	\subfigure{
			\includegraphics[width=6cm,height=6cm]{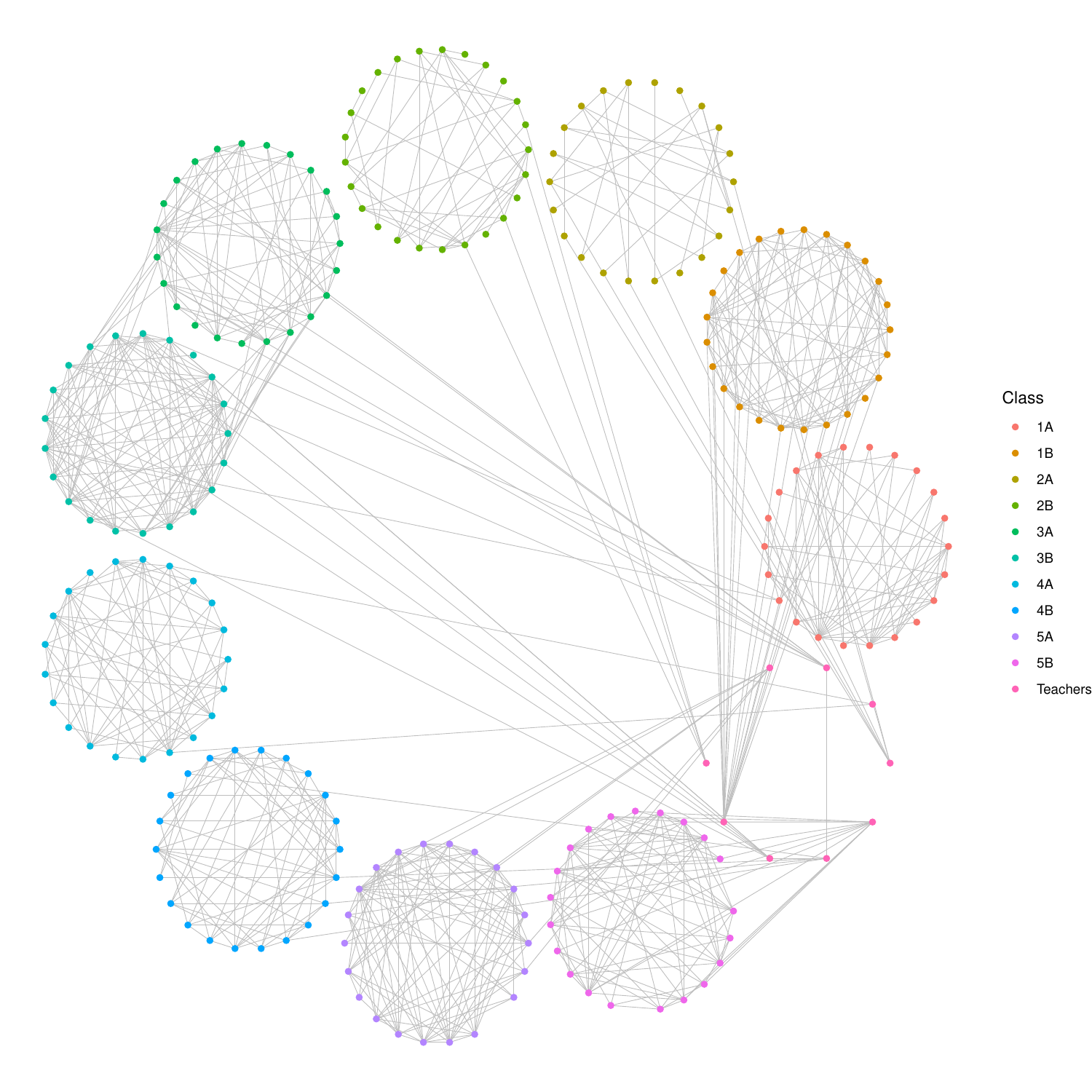}
		} 
		
		\caption{Aggregated networks for the first day. Note that there are 5 time periods in total (see Table~\ref{real data:2} for details), and they are visualized by rows.}
   \label{fig:realdata1}
	\end{figure}

	\section{Conclusion}
	\label{sec:conclusion}
	We propose a flexible approach for multiple change-point detection in Markov chain Bernoulli networks, which encompasses the models considered by \cite{wang2021optimal} and \cite{jiang2020autoregressive} as special cases. Through a novel clustering-based threshold, our method ``distills'' informative intervals from a sequence of random intervals. For networks with low-rank structure, we propose SUSVT for further refinement. Our method achieves almost a minimax detection bound and a minimax localization bound for heterogeneous networks with low-rank structure except a factor of logarithm, even as the number of change-points and network size diverge, without the need of sample splitting. A possible future work is to extend our RID and SUSVT to network sequences with other related dependent structures, such as network sequences with dependent and evolving latent structures, or dependence of subgraphs at different time points.
	
	\bibliographystyle{abbrv}

 \appendix
\section{Further descriptive analysis of the real data example}
\label{supp sec:real data}
For the real data example in Section~\ref{sec:real data}, we conduct Pearson's chi-squared tests for  \textit{$H_0: A_{kl}^{(i)}(t)$ and $A_{kl}^{(i)}(t-1)$ are independent}. We perform tests for all $k,l,i$ with the frequencies in the contingency tables larger than $5$ (i.e., $a_i\ge 5, i=1,2,3,4$ in Table~\ref{real data:1}). For illustration, we take $i=1, k=1, l=3$ and present the contingency table in Table~\ref{real data:1}. The $\chi^2$ statistic is $100.02$ and the $p$-value is fairly small. For all $k,l,i$, it turns out that more than $90\%$ of these tests demonstrate significant dependence (at the significance level of $0.05$ with Bonferroni corrections for multiple comparisons). To illustrate the independence in different individuals,  we select all pairs of individuals in $\left\{(i,j):\min_{k_1=1,2,k_2=0,1,k_3=0,1}\sum_t I\left(A_{i,j}^{(k_1)}(t)=k_2,A_{i,j}^{(k_1)}(t-1)=k_3\right)\ge 5\right\}$, and compute pairwise correlation coefficients.   Table~\ref{tab:sim1} summarises some quantiles of the correlation coefficients, and  the standard error is $0.044$. Table~\ref{tab:sim1}  indicates that the correlations between individuals are relatively weak.

\begin{table}[ht]
	\centering
	\caption{The contingency table for $\{A_{1,3}^{(1)}(t)\}$ which records the frequencies of $A_{1,3}^{(1)}(t-1)=j_1, A_{1,3}^{(1)}(t)=j_2$ for $j_1, j_2= 0,1$.}
	\label{real data:1}
	\begin{tabular}{cc|ccl}
		\multicolumn{2}{c|}{\multirow{2}{*}{Frequency}}       & \multicolumn{2}{c}{$A_{1,3}^{(1)}(t)$} &  \\ \cline{3-4}
		\multicolumn{2}{c|}{}                            & 0           & 1          &  \\ \cline{1-4}
		\multicolumn{1}{c|}{\multirow{2}{*}{$A_{1,3}^{(1)}(t-1)$}} & 0 & $a_1=354$         & $a_2=49$         &  \\
		\multicolumn{1}{c|}{}                        & 1 & $a_3=50$          & $a_4=65$         & 
	\end{tabular}
\end{table}

\begin{table}[ht]
	\centering
	\caption{Quantiles of the correlation coefficients between individuals, with the standard error being $0.044$.}
	\label{tab:sim1}
	\begin{tabular}{ccccccccc}
		\hline
		Cut point & 0.05   & 0.1    & 0.15   & 0.2    & 0.8   & 0.85  & 0.9   & 0.95  \\
		Quantile  & -0.092 & -0.075 & -0.066 & -0.059 & 0.062 & 0.090 & 0.131 & 0.211 \\ \hline
	\end{tabular}
\end{table}

\section{Figures for aggregated networks of interactions in real data example} 
\label{suppsec:Figures for aggregated networks of interactions in real data example}
The aggregated networks for the first day are presented in the main article, while the figures for the second day are shown in Figure~\ref{fig:realdata2}.

\begin{figure}[htbp]
	\centering 
	\subfigure{
		\includegraphics[width=6cm,height=6cm]{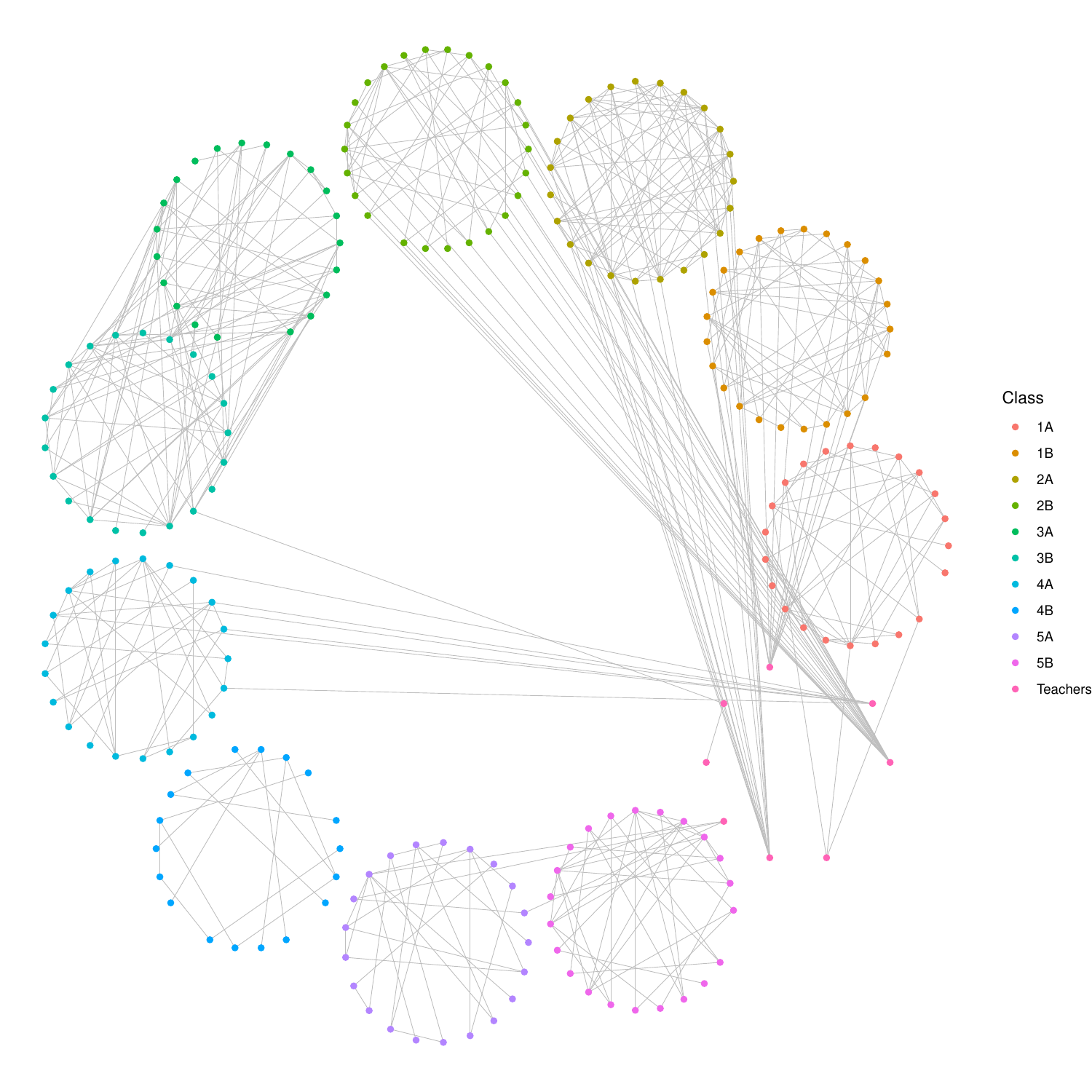}
	}\subfigure{
		\includegraphics[width=6cm,height=6cm]{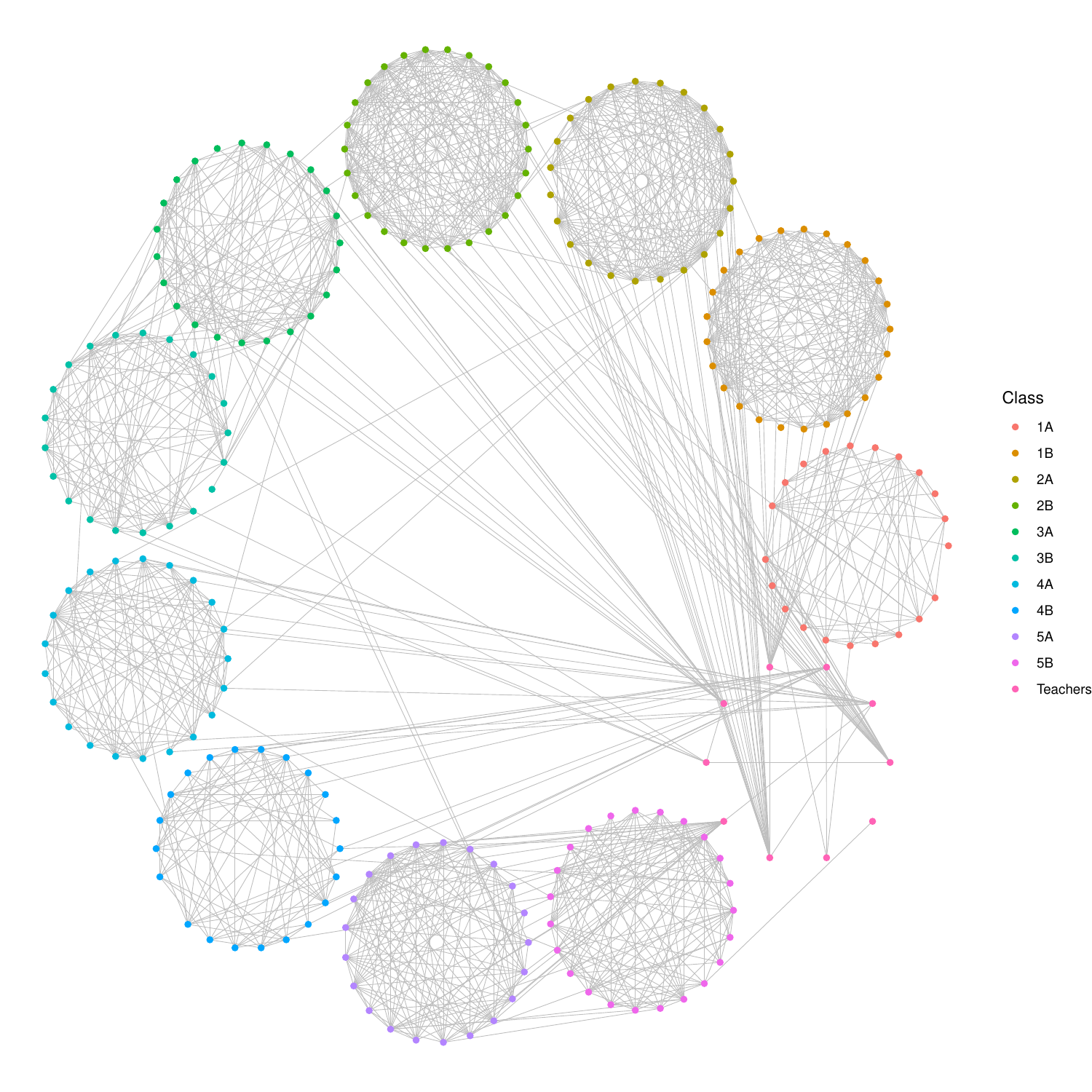}
	} 	\subfigure{
		\includegraphics[width=6cm,height=6cm]{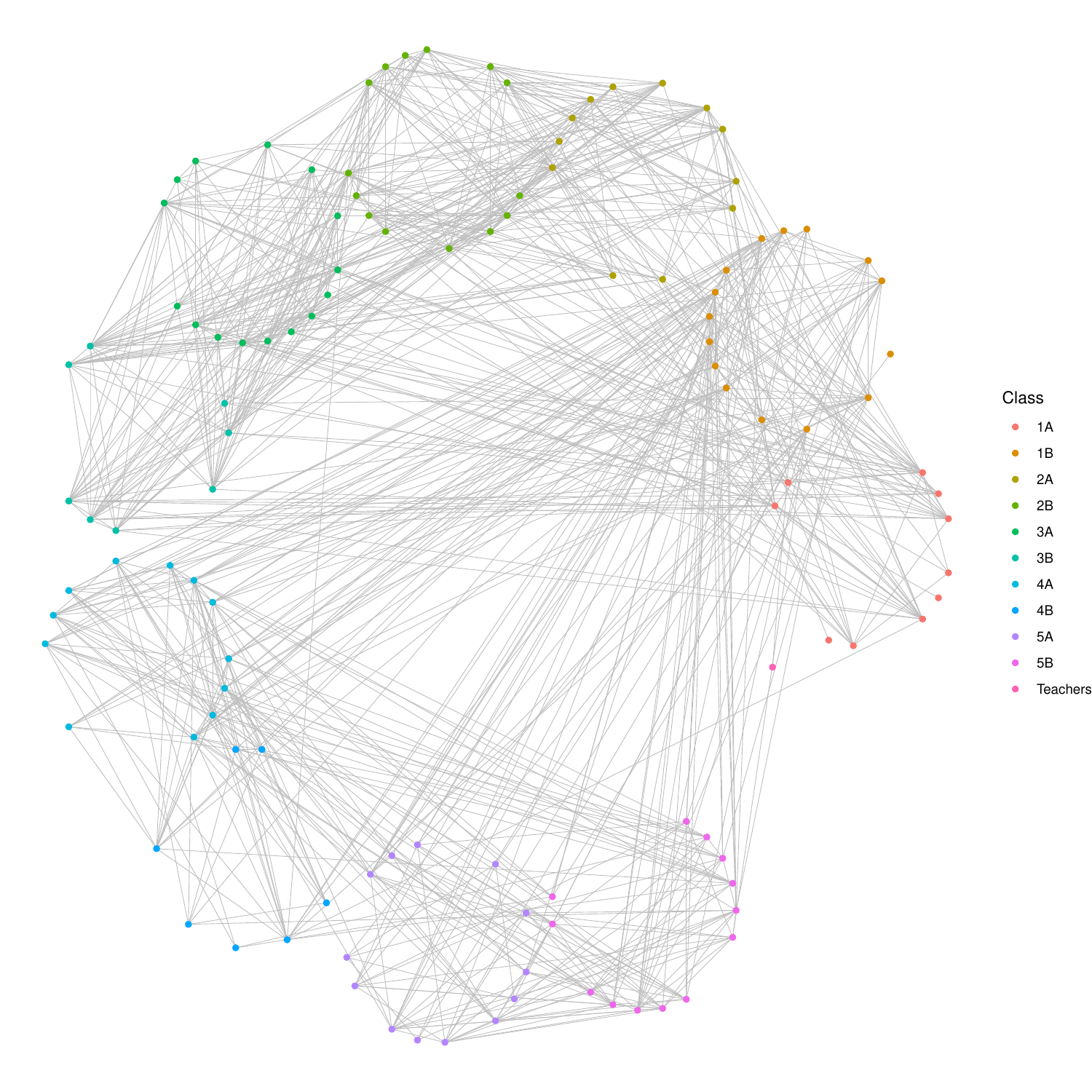}
	} 	\subfigure{
		\includegraphics[width=6cm,height=6cm]{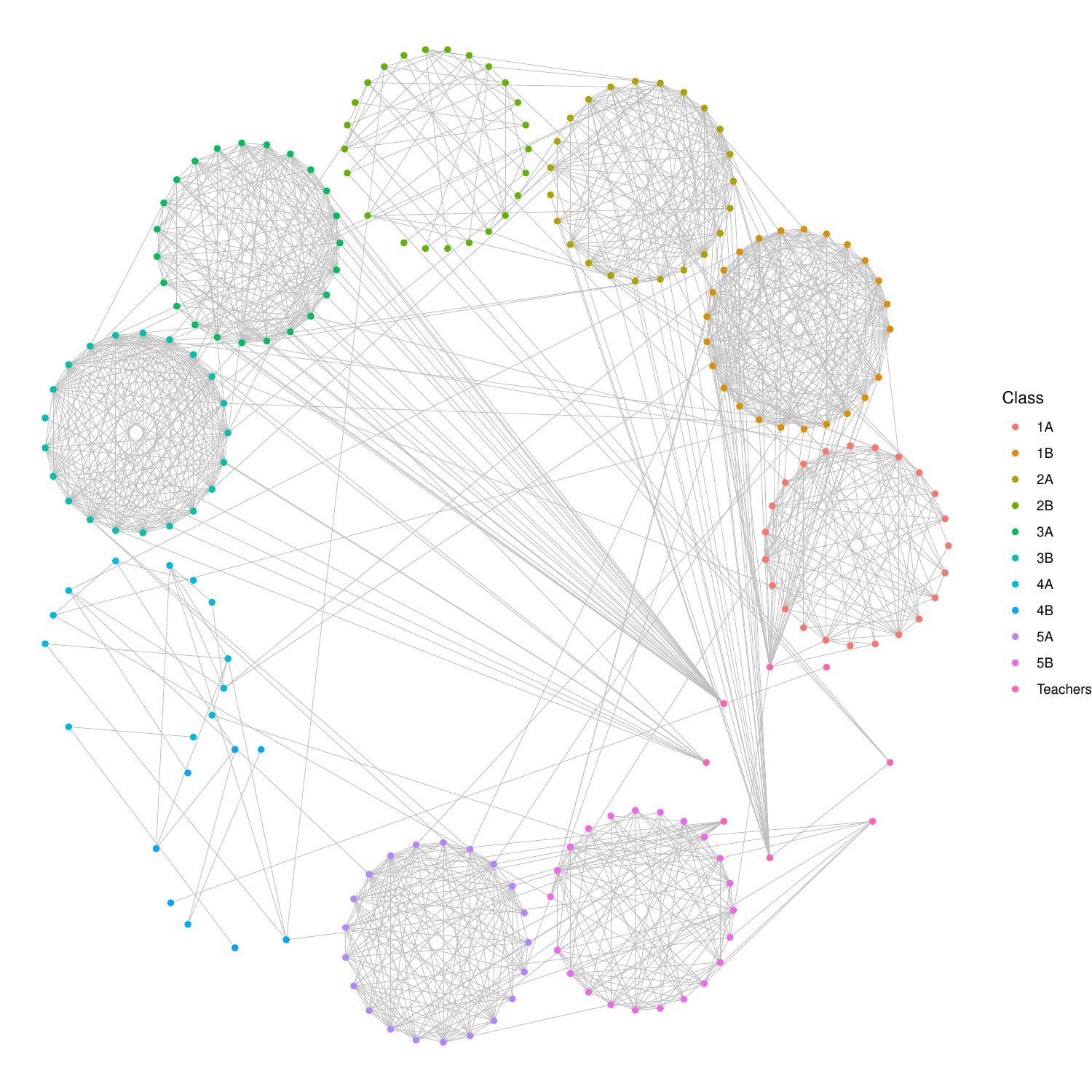}
	} 
	
	\caption{Aggregated networks for the second day. Note that there are 4 time periods in total (see Table~\ref{real data:2} for details), and they are visualized by rows.}
	\label{fig:realdata2}
\end{figure}

\section{Algorithm for universal singular value thresholding}
\label{suppsec:Algorithm for universal singular value thresholding}

We present the algorithm for universal singular value thresholding in Algorithm~\ref{alg:usvt}.

\begin{algorithm}[htbp]
	\caption{Universal Singular Value Thresholding in \citep{wang2021optimal}}
	\label{alg:usvt}
	\begin{algorithmic}[1]
		\Require Symmetric matrix $A \in \mathbb{R}^{n \times n}$, $\tau_2, \tau_4 > 0$.
		\State $(\kappa_i(A), v_i)$ be the $i_{th}$ eigen-pair of $A$, with $|\kappa_1(A)| \geq \cdots |\kappa_n(A)|$
		\State $A' = \sum_{i:|\kappa_i (A)|\ge\tau_2} \kappa_i(A) v_iv_i^{\top}$
		\State Let USVT$(A, \tau_2,\tau_4) = (A''_{ij})$ with
		\[
		(A'')_{ij} = \begin{cases}
			(A')_{ij}, & \text{if} \quad |(A'_{ij})| \le \tau_3\\
			\text{sign} ((A')_{ij})\tau_4, & \text{if} \quad |(A'_{ij})| > \tau_4\\
		\end{cases}
		\]
		\State Output USVT$(A, \tau_2,\tau_3)$.
	\end{algorithmic}
\end{algorithm}

\section{Proofs}
\label{suppsec:More proofs}

\begin{proof}[Proof of Theorem~\ref{th:K}]
	\label{pf:th1}
	
	For fixed $(i,j)$ and $s<t<e$, we have
	\begin{align*}
		(\tilde A_{ij})_{s,e}^t=\sqrt{\frac{e-t}{(e-s)(t-s)}}\sum_{r=s+1}^t A_{ij}(r)-\sqrt{\frac{t-s}{(e-s)(e-t)}}\sum_{r=t+1}^e A_{ij}(r).
	\end{align*}
	Then for another Markov chain $\{A_{ij}^*(t)\}$, we have $
	\left| (\tilde A_{ij})_{s,e}^t- (\tilde A_{ij}^*)_{s,e}^t \right|\le \sum_{r=s+1}^e |w_r||A_{ij}(r)-A_{ij}^*(r)|\le  \sum_{r=s+1}^e  |w_r|
	$ 
	where $
	w_r=\sqrt{\frac{e-t}{(e-s)(t-s)}}I(s+1\le r\le t)
	-\sqrt{\frac{t-s}{(e-s)(e-t)}}I(t+1\le r\le e)
	$. 
	Note that $\sum_{r=s+1}^e w_r=0,\sum_{r=s+1}^e w_r^2=1 $, then by Corollary 2.10 in \cite{paulin2015concentration}, we have that for any $\varepsilon>0$, 
	$
	\mathbb{P}\left(
	\left	|(\tilde A_{ij})_{s,e}^t-\mathbb{E}(\tilde A_{ij})_{s,e}^t\right|\ge \varepsilon
	\right)
	\le 2\exp\left({-2\varepsilon^2}/{
		\alpha
	}\right)
	$ where $\alpha=\inf_{0<\phi<1}\bar t_{mix}(\phi)\left(\frac{2-\phi}{1-\phi}\right)^2$ and $\bar t_{mix}$ is the mixing time for (inhomogeneous) Markov chain $\{A_{ij}(t)\}_{t=1}^T$ (see Definition 1.4 in \cite{paulin2015concentration}). Recall that we define 
	$\Xi = \sup_{i,j,t}|\mathbb{P}(A_{ij}(t)=1|A_{ij}(t-1)=1)-\mathbb{P}(A_{ij}(t)=1|A_{ij}(t-1)=0)| $ and assume $\Xi<1$, as a result, we have $\bar t_{mix}(\phi)\le -\log(\phi)/\log^*(1/\Xi)$ by \cite{samson2000concentration}. Therefore, we have $\alpha<20/(3\log^*(1/\Xi))$, and
	\begin{align}\label{eqn:1}
		\mathbb{P}\left(
		\left	|(\tilde A_{ij})_{s,e}^t-\mathbb{E}(\tilde A_{ij})_{s,e}^t\right|\ge \varepsilon
		\right)
		\le & 2\exp\left(\frac{-3\varepsilon^2\log^*(1/\Xi)}{
			10
		}\right).
	\end{align}
	
	By (\ref{eqn:1}) and Corollary 3.2 in \cite{bandeira2016sharp}, we have that $
	\sqrt{\log^*(1/\Xi)}
	\bb{E}|| 
	\tilde A_{s,e}^t-\mathbb{E}\tilde A_{s,e}^t
	||_{op}\le C_8^*(\sqrt{n}+\sqrt{\log(n)})\le C_8\sqrt{n}$ 
	where $C_8^*, C_8$ are two constants.
	Furthermore, by equation (2.8) and Corollary 4 in \cite{samson2000concentration}, we have $
	\bb{P}\left(\left|
	|| \tilde A_{s,e}^t ||_{op} - \bb{E}|| \tilde A_{s,e}^t ||_{op} 
	\right|\ge \varepsilon \right)\le 2\exp\left(
	-{\varepsilon^2}/{2}(1-\sqrt{\Xi})^2
	\right)$. Let the event
	\begin{equation}
		\label{eqn:ea}
		\mathcal{A}:=\left\{\forall s<t<e, \sup_{s<t<e} \left|
		||\tilde A_{s,e}^t||_{op}-||\bb{E}\tilde A_{s,e}^t||_{op} 
		\right|\le  \mathfrak{I}  \right\}
	\end{equation} 
	with $\mathfrak{I}=C(\sqrt{\log(T)}/(1-\sqrt{\Xi})+\sqrt{n}/\sqrt{\log^*(1/\Xi)})$ where $C$ is a sufficiently large constant.  Then we have  $\mathbb{P}(\mathcal A)\ge 1-T^{-1}$.  From (\ref{eqn:ea}) and Lemma~\ref{lem:lemma7}, we get $\mathbb{P}(\mathcal{A}\cap \mathcal{M})\ge 1-T^{-1}-\frac{T}{\Delta}\exp\left(-\frac{M\Delta^2}{32T^2} \right)$, 
	which will tend to $1$ as long as
	$ \frac{32T^2}{\Delta^2}log\left( \frac{T}{\Delta}\right)=o(M).$
	All the analysis in the rest of this proof is conducted on $\mathcal{A}\cap \mathcal{M}$. 
	
	For the random interval $(s_m,e_m)$, if it contains no change-point, $\mathbb{E}\tilde{A}_{s_m,e_m}^t=0$ for all $t\in(s_m,e_m)$. Then on event $\mathcal{A}$, we have $\max_t ||\tilde{A}_{s_m,e_m}^t||_{op}\le \mathfrak{I}$.  So the choice of $\tau>C(\sqrt{\log(T)}/(1-\sqrt{\Xi})+\sqrt{n}/\sqrt{\log^*(1/\Xi)})$ excludes intervals without change-points. Recall the construction of $S$ in Algorithm~\ref{alg:1}: ``If $\max_t||\tilde{A}_{s_m,e_m}^t||_{op}>\tau$; Then $S=S\cup \left\{(s_m,e_m]\right\}$; End If." 
	Therefore, the intervals in $S$ must include at least one change-point. We now prove that $\hat K=K$. First, it's obvious that $\hat K=\mathop{\max}\limits_{
		\left\{S^{'}\subset S : \text{intervals in $S^{'}$ are disjoint} \right\}
	}|S^{'}|$. 
	Since every interval in $S$ must cover at least one change-point, it follows from the Pigeonhole Principle that if we select more than $K$ intervals from $S$, then at least two intervals will cover the same change-point. Therefore, $|S^*|=\hat K\le K$. To prove $\hat K\ge K$, we construct a set $\mathcal{Q}$ such that $|\mathcal{Q}|=K, \mathcal{Q}\subset S$ and that intervals in $\mathcal{Q}$ are disjoint. Recall the definition of $\mathcal{M}$, for each $k=1,\cdots,K$, we can pick out an interval $(s_k,e_k)$ such that \begin{equation}
		\label{eqn:QQ}
		\eta_k-\frac{1}{4}\Delta < s_k < \eta_k-\frac{1}{8}\Delta,\ \eta_k+\frac{1}{8}\Delta < e_k < \eta_k+\frac{1}{4}\Delta .
	\end{equation}
	Let $\mathcal{Q}=\{(s_1,e_1),\cdots,(s_K,e_K)\}$. By the property of CUSUM statistics,\\ $
	\mathbb{E}\tilde A_{s_k,e_k}^{\eta_k}
	= \sqrt{\frac{(\eta_k-s_k)(e_k-\eta_k)}{e_k-s_k}}(\mathbb{E}A(\eta_k)-\mathbb{E}A(\eta_k-1))$.  Taking $||\cdot||_{op}$ for both sides to get that
	\begin{align*}
		||	\mathbb{E}\tilde A_{s_k,e_k}^{\eta_k}||_{op}&\ge 
		\sqrt{\frac{1}{2}[(\eta_k-s)\wedge (e-\eta_k)]}||\mathbb{E}A(\eta_k)-\mathbb{E}A(\eta_k-1)||_{op}
		\ge \frac{1}{4}\sqrt{\Delta}\kappa.
	\end{align*}
	Therefore, on event $\mathcal{A}$, we obtain $
	\max_t||\tilde{A}_{s_k,e_k}^t||_{op} \ge   ||\mathbb{E}\tilde A_{s_k,e_k}^{\eta_k}||_{op}-\mathfrak{I}\ge
	\frac{1}{4}\sqrt{\Delta}\kappa-\mathfrak{I}.
	$ 
	Therefore, $(s_k,e_k)\in S$ as long as $\tau<	\frac{1}{4}\sqrt{\Delta}\kappa-\mathfrak{I}$. 
	It suffices to show that intervals in $\mathcal{Q}$ are mutually disjoint. Notice that $e_k-\eta_k < \frac{1}{4}\Delta$ and $\eta_k-s_k < \frac{1}{4}\Delta$. Meanwhile, $\Delta \le \eta_k-\eta_{k-1}$ for all $k$. This indicates that all $(s_k,e_k)$ are mutually disjoint. Therefore, $\hat K=K$. 
	
	Recall the construction of $r_k$ and $l_k$ in Line 10 and 17 of Algorithm~\ref{alg:1}. Since $\eta_1<r_1\le e_1$, $e_1<s_2$ and every interval covers at least one change-point, we conclude that $\eta_2<r_2\le e_2$. It follows by induction that $\eta_k<r_k\le e_k$ for $1\le k\le K$. By symmetry, $s_k\le l_{k}<\eta_k$ for $1\le k\le K$. So $\left[l_k,r_{k}\right]$ contains $\eta_{k}$ only and $r_{k}-l_k\le e_{k}-s_{k}\le \frac{\Delta}{2}$.
	
\end{proof}

\begin{proof}[Proof of Theorem~\ref{th:Theoretical result of local refinement}] 
	
	By Theorem~\ref{th:K}, after applying Algorithm~\ref{alg:1}, we have $\mathbb{P}(\hat K=K, \mathcal{B})\to 1$ where $\mathcal{B}:=\bigcap_{j=1}^{\hat K}\left\{ \left[l_{j},r_{j}\right]  \text{ covers } \eta_j \text{ and }r_{j}-l_{j}\le {\Delta}/{2} \right\}$. 
	In the following proof, when there is no ambiguity, we omit the subscripts of $s_{k}$ and $e_{k}$ and write them as $s$ and $e$. 
	We start by outlining our proof. In Step 1, we show that the $s,v_k,e$ in Line 5 of Algorithm~\ref{alg:3}  are well-positioned in that the spacing between them are neither too far nor too close. In Step 2, We show that the Frobenius norm of the expectation of  $\tilde Y_{s,e}^{v_k}$ in Line 10 of Algorithm~\ref{alg:3}  is of order $\frac{\kappa_{\eta_k}^2\Delta}{\tau_3\log(T)}$. In Step 3, we derive the bound of $||\tilde Y_{s^*,e^*}^t-\mathbb{E}\tilde Y_{s^*,e^*}^t||_{op}$ for any $(s^*,e^*)\subset (0,T)$ and explore the effect the  of USVT. In Step 4, we show that the samples drawn at each $\tau_3\log(T)$ are nearly independent, and obtained the concentration results required by Step 5. With these concentration inequalities and the results of Step 3, we finish the proof via arguments motivated by \cite{wang2021optimal}.

	\textbf{Step 1.}  We fix for some $k (1\le k\le \hat K)$. Notice that $\tau_3\log(T)=o(\Delta)$, we can assume without loss of generality that $\frac{v_k-s-1}{2\tau_3\log(T)}, \frac{e-s-1}{2\tau_3\log(T)}$ are positive integers. Assume without loss of generality that $l_k-\hat\Delta/16, r_k+\hat\Delta/16, (l_k+r_k)/2$ are positive integers. 
	Recall that $S^*=\{[l_k,r_k]\}_{k=1}^{\hat K}$ and $r_k-l_k\le \frac{\Delta}{2}$. Since $|\eta_k-v_k|\le \frac{\Delta}{4}$, we have $\frac{1}{2}\Delta\le \hat\Delta\le \frac{3}{2}\Delta$. $e-\eta_k\ge e-r_k\ge \frac{\Delta}{32}.$ Similarly, $\eta_k-s\ge \frac{\Delta}{32}, e-s\ge\frac{\Delta}{16}, e-s\le\frac{3\Delta}{4}, e-v_k\ge \frac{\Delta}{32}, e-v_k\le \frac{\Delta}{2}, v_k-s\ge \frac{\Delta}{32}$. 
	
	\textbf{Step 2.} Let 
	$\kappa_{\eta_k}=||\mathbb{E}A(\eta_k)-\mathbb{E}A(\eta_k-1)||_F$. It's obvious that $\frac{e-s-1}{\tau_3\log(T)}+1\ge \frac{\Delta}{16\tau_3\log(T)}$. As a result, $\tilde\Delta_k^2\ge \frac{1}{2}\min\left(\frac{v_k-s-1}{2\tau_3\log(T)}+1,\frac{e-v_k}{2\tau_3\log(T)}\right)\ge \frac{1}{4}\left(\frac{\Delta/32}{2\tau_3\log(T)}\right)\ge \frac{\Delta}{128\tau_3\log(T)}$. 
	Assume without loss of generality that $v_k\le \eta_k$. By some direct calculations and the previous inequality, we have
	\begin{equation}
		\label{pfeqn:0}
		||\mathbb{E}\tilde Y_{s,e}^{v_k}||_F^2=\tilde\Delta_k^2\left(\frac{e-\eta_k}{e-v_k}\right)^2\kappa_{\eta_k}^2\ge \frac{\Delta}{128\tau_3\log(T)}\frac{1}{16^2}\kappa_{\eta_k}^2=\frac{\kappa_{\eta_k}^2\Delta}{32768\tau_3\log(T)}.
	\end{equation}
	
	\textbf{Step 3.} Let $Y(t,k)=A(s_k+1+2\tau_3\log(T)(t-1)), 1\le t\le 1+\frac{e_k-s_k-1}{2\tau_3\log(T)}$. We abbreviate $Y(t,k)$ as $Y(t)$ when there is no ambiguity. Let  $\bar Y(t)=Y(t)-\mathbb{E}Y(t)$.  For any $1\le i,j\le n, Y_{ij}(t)$ is a Markov chain. We have \begin{equation}
		\label{pfeqn:8}
		||\mathbb{P}(\cdot|Y_{ij}(t-1)=1)-\mathbb{P}(\cdot|Y_{ij}(t-1)=0)||_{TV}\le \Xi^{2\tau_3\log(T)}\le T^{-2c_7}\le \frac{1}{4}.
	\end{equation}
	Then equation (2.8) in \cite{samson2000concentration} becomes $||\Gamma||^2\le \left(\frac{1}{1-\sqrt{\Xi^{2\tau_3\log(T)}}}\right)^2\le 4$. Let $w_t$ be a sequence of weights that satisfy $\sum w_t=0, \sum w_t^2=1$. Since $\sum_t w_t \bar {Y}_{ij}(t)$ and $||\sum_t w_t \bar {Y}(t)||_{op}$ are both 1-Lipschitz. By Corollary 4 of \cite{samson2000concentration}, for any $\varepsilon>0$, we have 
	\begin{equation}
		\label{pfeqn:1}
		\mathbb{P}\left(\|\sum_t w_t \bar {Y}(t)\|_{op}\ge \mathbb{E}\|\sum_t w_t \bar {Y}(t)\|_{op}+2\varepsilon\right)\le e^{-\frac{\varepsilon^2}{2}},
	\end{equation}
	\begin{equation}
		\label{pfeqn:2}
		\mathbb{P}\left(\left|\sum_t w_t \bar {Y}_{ij}(t)\right|\ge 2\varepsilon\right)\le 2e^{-\frac{\varepsilon^2}{2}}.
	\end{equation}
	By Definition~\ref{def:Bernoulli network Markov chain}, $\{\sum_t w_t \bar {Y}_{ij}(t)\}_{1\le i,j\le n}$ are independent. By (\ref{pfeqn:2}) and Corollary 3.2 in \cite{bandeira2016sharp}, we get
	\begin{equation}
		\label{pfeqn:3}
		\mathbb{E}\|\sum_t w_t \bar {Y}(t)\|_{op}\le C_8^* (\sqrt{n}+\sqrt{\log n})\le C_8\sqrt{n},
	\end{equation}
	where $C_8^*, C_8>8$ are two absolute constants. Combine (\ref{pfeqn:3}) with (\ref{pfeqn:1}) and set $\varepsilon=\frac{C_8}{2}\sqrt{\log(T)}$, then we have
	\begin{equation}
		\label{pfeqn:9}
		\mathbb{P}\left(\|\sum_t w_t \bar {Y}(t)\|_{op}\ge C_8(\sqrt{n}+\sqrt{\log(T)})\right)\le T^{-8}.
	\end{equation}
	Define the event $\mathcal{A}=\left\{\sup_{0\le s^*<t\le e^*\le T}\|\tilde Y_{s^*,e^*}^t-\mathbb{E}\tilde Y_{s^*,e^*}^t\|_{op}\le C_8(\sqrt{n}+\sqrt{\log(T)})\right\}$, by the union bound argument and (\ref{pfeqn:9}), we have $\mathbb{P}(\mathcal{A})\ge 1-T^{-5}$. Let $\tau_2=\frac{4}{3}C_8(\sqrt{n}+\sqrt{\log(T)})$. On event $\mathcal{A}$, by Lemma 1 (applying it with $\delta=\frac{1}{3}$) in \cite{xu2017rates}, we have \[
	\sup_{0\le s^*<t\le e^*\le T}||\text{USVT}(\tilde Y_{s^*,e^*}^t,\tau_2,\infty)-\mathbb{E}\tilde Y_{s^*,e^*}^t||_{F}\le 4\tau_2\sqrt{r}.\]
	Let $\mathcal{D}$ denote the above event. Similar to (\ref{pfeqn:0}), we have 
	\begin{equation}
		\label{pfeqn:6}
		\mathbb{E} (\tilde Y_{s^*,e^*}^{v_k})_{ij}=\tilde\Delta_k\left(\frac{e^*-\eta_k}{e^*-v_k}\right)\le \tilde\Delta_k.
	\end{equation}
	Therefore, for any $1\le i,j\le n$, we have \[| (\text{USVT}(\tilde Y_{s,e}^{v_k},\tau_2,\tilde\Delta_k))_{ij}-(\mathbb{E}\tilde Y_{s,e}^{v_k})_{ij} |\le |(\text{USVT}(\tilde Y_{s,e}^{v_k},\tau_2,\infty))_{ij}-(\mathbb{E}\tilde Y_{s,e}^{v_k})_{ij}|.\]
	So on event $\mathcal{D}$, we have $||\hat Y_k-\mathbb{E}\tilde Y_{s,e}^{v_k}||_{F}\le ||\text{USVT}(\tilde Y_{s,e}^{v_k},\tau_2,\infty)-\mathbb{E}\tilde Y_{s,e}^{v_k}||_{F} \le 4\tau_2\sqrt{r}$. 
	By Assumption~\ref{assumption:local refinement} and (\ref{pfeqn:0}), we have 
	\begin{equation}
		\label{pfeqn:7}
		||\hat Y_k||_F\ge ||\mathbb{E}\tilde Y_{s,e}^{v_k}||_{F}-4\tau_2\sqrt{r}\ge \frac{\kappa_{\eta_k}\sqrt\Delta}{200\sqrt{\tau_3\log(T)}}
	\end{equation}
	when $T$ is large enough. As a result, similarly to the derivations in \cite{wang2021optimal}, 
	we have when $T$ is large enough, $\left< {\mathbb{E}\tilde Y_{s,e}^{v_k}}/{||\mathbb{E}\tilde Y_{s,e}^{v_k}||_F}, {\hat Y_k}/{||\hat Y_k||_F}\right>\ge \frac{1}{2}$. 
	Therefore, \begin{equation}
		\label{pfeqn:5}
		\left<\mathbb{E}\tilde Y_{s,e}^{v_k},\frac{\hat Y_k}{||\hat Y_k||_F}\right>\ge \frac{1}{2}||\mathbb{E}\tilde Y_{s,e}^{v_k}||_F\ge \frac{\kappa_{\eta_k}\sqrt \Delta}{400\sqrt{\tau_3\log(T)}}.
	\end{equation}
	
	\textbf{Step 4.} Let $Z(t)=A(s+1+\tau_3\log(T)(2t-1)), 1\le t\le \frac{e-s-1}{2\tau_3\log(T)}$. By Lemma~\ref{lemma:markov conditional independent}, Conditional on $\{Y(t)\}$, we get that $\{Z(t)\}$ are independent. By Bernstein inequality, for any $\varepsilon>0$, 
	\begin{align*}
		&\mathbb{P}\left(\left|\frac{1}{\sqrt{
				\frac{e-s-1}{2\tau_3\log(T)}
		}}\sum_{t=1}^{\frac{e-s-1}{2\tau_3\log(T)}}
		\left<\mathbb{E}(Z(t)|\{Y(t)\})-Z(t),\frac{\hat Y_k}{||\hat Y_k||_F}\right>\right|\ge \varepsilon \Bigg|\{Y(t)\}\right)\\
		& \le 2\exp\left(
		-\frac{1.5\varepsilon^2}{3+\varepsilon \frac{1}{\sqrt{\frac{e-s-1}{2\tau_3\log(T)}}}\frac{||\hat Y_k||_\infty}{||\hat Y_k||_F}}
		\right).
	\end{align*}
	By (\ref{pfeqn:6}), (\ref{pfeqn:7}), we get that on event $\mathcal{A}$, 
	${||\hat Y_k||_\infty}/{||\hat Y_k||_F}/{\sqrt{\frac{e-s-1}{2\tau_3\log(T)}}}\le {200\sqrt{\tau_3\log(T)}}/(\kappa_{\eta_k}\sqrt\Delta)$. 
	Notice that the right hand side of the above inequality is $o(1)$ because of Assumption~\ref{assumption:local refinement}. Setting $\varepsilon=3\log(T)$, we obtain that when $T$ is large enough, 
	\[\mathbb{P}\left(\left|\frac{1}{\sqrt{
			\frac{e-s-1}{2\tau_3\log(T)}
	}}\sum_{t=1}^{\frac{e-s-1}{2\tau_3\log(T)}}\left<\mathbb{E}(Z(t)|\{Y(t)\})-Z(t),\frac{\hat Y_k}{||\hat Y_k||_F}\right>\right|\ge 3\log(T) \Bigg| \mathcal{A}\right)\le 2T^{-4.5},
	\]
	\begin{equation}
		\label{pfeqn:10}
		\mathbb{P}\left(\left|\frac{1}{\sqrt{
				\frac{e-s-1}{2\tau_3\log(T)}
		}}\sum_{t=1}^{\frac{e-s-1}{2\tau_3\log(T)}}
		\left<\mathbb{E}(Z(t)|\{Y(t)\})-Z(t),\frac{\hat Y_k}{||\hat Y_k||_F}\right>\right|\ge 3\log(T)\right)\le 2T^{-4.5}.
	\end{equation}
	We then bound \[
	I_1=\left|\frac{1}{\sqrt{
			\frac{e-s-1}{2\tau_3\log(T)}
	}}\sum_{t=1}^{\frac{e-s-1}{2\tau_3\log(T)}}
	\left<\mathbb{E}(Z(t)|\{Y(t)\})-\mathbb{E}Z(t),\frac{\hat Y_k}{||\hat Y_k||_F}\right>\right|.
	\]
	By the derivation of (\ref{pfeqn:8}), we can use Lemma~\ref{lemma:markov tv conditional} with $\Xi^*=T^{-c_7}$ and the Cauchy-Schwarz inequality, so that
	\begin{align}
		\label{pfeqn:11}
		I_1&= \frac{1}{\sqrt{
				\frac{e-s-1}{2\tau_3\log(T)}
		}}\left|\left<\sum_{t=1}^{\frac{e-s-1}{2\tau_3\log(T)}}\left(\mathbb{E}(Z(t)|\{Y(t)\})-\mathbb{E}Z(t)\right),\frac{\hat Y_k}{||\hat Y_k||_F}\right>\right|\\
		&\le \frac{1}{\sqrt{
				\frac{e-s-1}{2\tau_3\log(T)}
		}}
		\left\|\sum_{t=1}^{\frac{e-s-1}{2\tau_3\log(T)}}\left(\mathbb{E}(Z(t)|\{Y(t)\})-\mathbb{E}Z(t)\right)\right\|_F\notag\\
		&= \frac{1}{\sqrt{
				\frac{e-s-1}{2\tau_3\log(T)}
		}} \sqrt{
			\sum_{i,j=1}^n \left(\sum_{t=1}^{\frac{e-s-1}{2\tau_3\log(T)}}(\mathbb{E}(Z_{ij}(t)|\{Y(t)\})-\mathbb{E}Z_{ij}(t)) \right)^2
		}\notag\\
		&\le \frac{1}{\sqrt{
				\frac{e-s-1}{2\tau_3\log(T)}
		}} \sqrt{
			\sum_{i,j=1}^n \left(3\Xi^*\frac{e-s-1}{2\tau_3\log(T)} \right)^2
		}=
		3n\Xi^* \sqrt{
			\frac{e-s-1}{2\tau_3\log(T)}
		}.\notag
	\end{align}
	Recall that $n\sqrt{\Delta}T^{-c_7}\le 1$,  combining (\ref{pfeqn:10}) with (\ref{pfeqn:11}),  we have that
	\[\mathbb{P}\left(\left|\frac{1}{\sqrt{
			\frac{e-s-1}{2\tau_3\log(T)}
	}}\sum_{t=1}^{\frac{e-s-1}{2\tau_3\log(T)}}
	\left<\mathbb{E}Z(t)-Z(t),\frac{\hat Y_k}{||\hat Y_k||_F}\right>\right|\ge \sqrt{
		\frac{3}{2\tau_3\log(T)}
	}+3\log(T)\right)\le 2T^{-4.5}.
	\]
	Since $\tau_3\ge 1/(-\log(\exp(-1)))\ge 1$, we have that 
	\begin{equation}
		\label{pfeqn:12}
		\mathbb{P}\left(\left|\frac{1}{\sqrt{
				\frac{e-s-1}{2\tau_3\log(T)}
		}}\sum_{t=1}^{\frac{e-s-1}{2\tau_3\log(T)}}
		\left<\mathbb{E}Z(t)-Z(t),\frac{\hat Y_k}{||\hat Y_k||_F}\right>\right|\ge 4\log(T)\right)\le 2T^{-4.5}
	\end{equation}
	when $T$ is large enough. Similar arguments  show that for $s+(e-s)/100<t\le e-(e-s)/100$,
	\begin{equation}
		\label{pfeqn:13}
		\mathbb{P}\left(\left|
		\left<\mathbb{E}\tilde Z_{s,e}^t-\tilde Z_{s,e}^t,\frac{\hat Y_k}{||\hat Y_k||_F}
		\right>\right|\ge 4\log(T)\right)\le 2T^{-4.5}.
	\end{equation}
	
	\textbf{Step 5.} Consider the one dimensional time series $a(t)=<Z(t),\frac{\hat Y_k}{||\hat Y_k||_F}>$. By Lemma S.2.4 in the supplemental material of \cite{wang2017optimal}, using (\ref{pfeqn:5}), (\ref{pfeqn:12}) and (\ref{pfeqn:13}), we have that $|b_k^*-\eta_k|\le C_{10}\frac{\log^2(T)}{\kappa_{\eta_k}^2}$, and $\sup_k|b_k^*-\eta_k|\le C_{10}\frac{\log^2(T)}{ \kappa_2^2}$, where $b_k^*$ is $\arg\max_t<\tilde Z(t),\hat Y_k>$ and $\tilde Z(t)$ is the CUSUM statistic on $Z(t)$. Recall that $Z(t)=A(s+1+\tau_3\log(T)(2t-1))$. So we get 
	$\sup_k|\hat\eta_k^*-\eta_k|\le C_{4}^*{c_7\log^3(T)}/(\log^*(1/\Xi)\kappa_2^2)$ when $T$ is large enough.  Then the result follows by noticing that $\sup_k|\hat\eta_k-\eta_k|=O(\sup_k|\hat\eta_k^*-\eta_k|)$.
\end{proof}

\begin{proof}[Proof of Theorem~\ref{th:network refinement sbm}]
	By tracking the proof of Theorem~\ref{th:Theoretical result of local refinement}, it suffices to show that Step 3 in the proof of Theorem~\ref{th:Theoretical result of local refinement} still holds. Moreover, the entire derivation is similar to the proof of Theorem 3 in \cite{wang2021optimal}, so we omit the proof.
\end{proof}

\begin{proof}[Proof of Theorem~\ref{the:gmm}]
	\label{pf:th2}
	To simplify the notations, throughout this proof we  set $\eta_0=1$ and $\eta_{K+1}=T+1$. 
	We first prove that under the conditions of Theorem~\ref{the:gmm}, we have
	\begin{enumerate}
		\item $\bb E\left(\max_t||\bb{E}\tilde{A}_{s_m,e_m}^t||_{op}\bigg|\{|(s_m,e_m]\cap \eta|=0\},\mathcal{A}\right)=0$.
		\item $\bb E\left(\max_t||\bb{E}\tilde{A}_{s_m,e_m}^t||_{op}\bigg|\{|(s_m,e_m]\cap \eta|=0\},\mathcal{A}\right)>b$.
		\item			$\bb P\left(\bigcup_{m=1}^M\left\{0<\max_t||\bb{E}\tilde{A}_{s_m,e_m}^t||_{op}< b\right\}\bigg|\mathcal{A}\right)\to 0.$
	\end{enumerate}
	Then we prove the theorem itself.

	Since $\{(s_l,e_l]\}_{l=1}^M$ are i.i.d. sampled, it suffices to consider one particular random $(s_m,e_m]$ for some $m$. 
	Let
	$$
	s_m+1<\eta_r<\eta_{r+1}<\cdots<\eta_{r+q^*} \le e_m
	$$
	where $\eta_k, k=r,\cdots,r+q^*$ are change-points and $q^*\ge -1$ ($q^*= -1$ means there is no change-point in $(s_m,e_m]$). 
	Notice that $q^*$ is determined by $s_m$ and $e_m$, so the above sets are random. 
	Define events
	$$\mathcal{A}_1=\{q^*=-1\},
	\mathcal{A}_2^\gamma=\bigcup_{1\le r\le K}\{q^*=0, (\eta_r-s_m-1)\wedge(e_m+1-\eta_r)\le \gamma\},$$
	$$\mathcal{A}_3^\gamma=\bigcup_{1\le r\le K-1}\{q^*=1, (\eta_r-s_m-1)\vee(e_m+1-\eta_{r+1})\le \gamma\},$$
	where $\gamma<\Delta$. 
	$\mathcal{A}_1$ means that there is no change-point in the interval $(s_m,e_m]$.  $\mathcal{A}_2^\gamma$ indicates that there is only one change-point in $(s_m,e_m]$ and this change-point is close to at least one of the  endpoints. $\mathcal{A}_3^\gamma$ indicates that there are two change-points in $(s_m,e_m]$ and both of them are close to the endpoints.  
	Since the interval $(s_m,e_m]$ is sampled uniformly, we can directly calculate the probabilities of some of these events.  
	For $\mathcal{A}_1$, we have
	\begin{equation}
		\label{eqn:PA1}
		\mathbb{P}(\mathcal{A}_1)=\sum_{k=1}^{K+1} \left(\frac{\eta_k-\eta_{k-1}}{T} \right)^2=\frac{1}{T^2}\sum_{k=1}^{K+1}(\eta_k-\eta_{k-1})^2.
	\end{equation}
	To bound (\ref{eqn:PA1}), we observe that
	\[\frac{1}{T^2}\sum_{k=1}^{K+1}(\eta_k-\eta_{k-1})^2\ge \frac{1}{T^2}\frac{1}{K+1}\left(\sum_{k=1}^{K+1}(\eta_k-\eta_{k-1})\right)^2=\frac{1}{K+1}.\]
	For the upper bound, since $\sum_{k=1}^{K+1}(\eta_k-\eta_{k-1})=T$ and $\eta_k-\eta_{k-1}\ge \Delta$, we derive that
	\begin{align*}
		\frac{1}{T^2}\sum_{k=1}^{K+1}(\eta_k-\eta_{k-1})^2&=\frac{1}{T^2}\left(\sum_{k=1}^{K+1}(\eta_k-\eta_{k-1}-\Delta)^2-(K+1)\Delta^2+2\Delta T\right)\\
		&\le \frac{1}{T^2}\left(\left(\sum_{k=1}^{K+1}(\eta_k-\eta_{k-1})-(K+1)\Delta\right)^2-(K+1)\Delta^2+2\Delta T\right)\\
		&= \frac{1}{T^2}\left((T-(K+1)\Delta)^2-(K+1)\Delta^2+2\Delta T\right)\\
		&= \frac{1}{T^2}(K\Delta^2+(T-K\Delta)^2).
	\end{align*}
	Therefore, 
	\begin{equation}
		\label{ieqn:PA1}
		\frac{1}{K+1}\le   \mathbb{P}(\mathcal{A}_1) \le \frac{1}{T^2}(K\Delta^2+(T-K\Delta)^2) .
	\end{equation}
	As for $\mathcal{A}_2^\gamma$, we have 
	\begin{equation}
		\label{eqn:PA2}
		\mathbb{P}(\mathcal{A}_2^\gamma)=\frac{\gamma}{T^2}\sum_{k=1}^{K}(\eta_{k+1}-\eta_{k-1})\le \frac{2\gamma}{T}.
	\end{equation}
	For $\mathcal{A}_3^\gamma$, similarly we have
	\begin{equation}
		\label{eqn:PA3}
		\mathbb{P}(\mathcal{A}_3^\gamma)=\frac{\gamma^2}{T^2}(K-1).
	\end{equation}

	\textbf{Proof of (a).} 
	On $\mathcal{A}_1$, i.e., $|(s_m,e_m]\cap\eta|=0$, from proof of Theorem \ref{th:K} and the fact that $\mathbb{E}\tilde{A}_{s_m,e_m}^t=0$, we have $\max_t||\mathbb{E}\tilde{A}_{s_m,e_m}^t||_{op}=0$, \[\bb E\left(\max_t||\bb{E}\tilde{A}_{s_m,e_m}^t||_{op} \bigg|\{|(s_m,e_m]\cap \eta|=0\},\mathcal{A}\right)=0.\]
	
	\textbf{Proof of (b).}  On the event $\{q^*\ge 1\}\cap ( \mathcal{A}_3^\gamma)^c$, there exists a change-point $\eta_k$ such that any other change-point $\eta_{k^{'}}$ and the endpoints $s_m,e_m$ satisfy $\min\{\eta_k-s_m-1,e_m+1-\eta_k,|\eta_k-\eta_{k^{'}}|\}\ge \gamma$. By Lemma \ref{lem:lem8}, $$
	\max_t||\bb{E}\tilde{A}_{s_m,e_m}^t||_{op}
	\ge  \frac{1}{4\sqrt{e_m-s_m}}\kappa\gamma\ge \frac{1}{4\sqrt{T}}\kappa\gamma .
	$$
	Recall that $\tilde{g}_{s_m,e_m}^t$ is the CUSUM statistic, therefore, on the event $\{q^*= 0\}\cap ( \mathcal{A}_2^{\gamma^2/T})^c$, direct calculations show that
	\begin{align*}
		\max_t ||\mathbb{E}\tilde{A}_{s_m,e_m}^t||_{op} \ge \sqrt{\frac{\gamma^2/T}{2}}\kappa\ge \frac{1}{4\sqrt{T}}\kappa\gamma .
	\end{align*}
	As a result, on the event $\mathcal{A}_4:=\{q^*\ge0\}\cap ( \mathcal{A}_2^{\gamma^2/T})^c\cap( \mathcal{A}_3^\gamma)^c$, we have $\max_t ||\mathbb{E}\tilde{A}_{s_m,e_m}^t||_{op}\ge \frac{1}{4\sqrt{T}}\kappa \gamma .$ We derive that  
	\begin{align*}
		&\bb E\left(\max_t||\mathbb{E}\tilde{A}_{s_m,e_m}^t||_{op}\bigg|\{|(s_m,e_m]\cap \eta|>0\},\mathcal{A}\right)\\
		&\ge \frac{1}{4\sqrt{T}}\kappa\gamma \frac{\bb P\left(\{q^*\ge0\}\cap ( \mathcal{A}_2^{\gamma^2/T})^c\cap( \mathcal{A}_3^\gamma)^c\right)}{\bb P(q^*\ge0)}\\
		&\ge \frac{1}{4\sqrt{T}}\kappa\gamma \left(1-\frac{2\gamma^2/T^2+\gamma^2(K-1)/T^2}{1-\frac{1}{T^2}(K\Delta^2+(T-K\Delta)^2)
		}\right)\\
		&\ge  \frac{1}{4\sqrt{T}}\kappa\gamma \left(1-\frac{\gamma^2(K+1)}{K\Delta(2T-(K+1)\Delta)}\right)\\
		&\ge  \frac{1}{4\sqrt{T}}\kappa\gamma \left(1-\frac{2\gamma^2}{\Delta T}\right).
	\end{align*}
	We set $\gamma=\frac{8b\sqrt{T}}{\kappa}$ which is of order $o(\Delta)$ since $b=o(\kappa\Delta/\sqrt{T})$.  Therefore, \\$\bb E\left(\max_t||\mathbb{E}\tilde{A}_{s_m,e_m}^t(X)||_{op}\bigg|\{|(s_m,e_m]\cap \eta|>0\},\mathcal{A}\right)\ge \frac{1}{6\sqrt{T}}\kappa\gamma>b$. 
	
	\textbf{Proof of (c).} Define the event 
	\[\mathcal{A}_5:=\left\{(\mathcal{A}_2^{\gamma^2/T})^c\cap( \mathcal{A}_3^\gamma)^c\text{ for all } 1\le m\le M\right\}.\] Note that $( \mathcal{A}_2^{\gamma^2/T})^c\cap( \mathcal{A}_3^\gamma)^c=\left(\{q^*\ge0\}\cap ( \mathcal{A}_2^{\gamma^2/T})^c\cap( \mathcal{A}_3^\gamma)^c\right)\cup \{q^*=-1\}$. 
	Then $\bb P(\mathcal{A}_5)\ge \left(1-\frac{\gamma^2(K+1)}{T^2}\right)^M\ge 1-\frac{M\gamma^2(K+1)}{T^2}\to 1$ since $\gamma=o(T/\sqrt{KM})$ by the above construction. On $\mathcal{A}_5$, we have $\max_t||\bb E\tilde{A}_{s_m,e_m}^t||_{op}\in 0\cup [b,\infty)$  for all $1\le m\le M$. Therefore, 
	\begin{align*}
		\bb P\left(\bigcup_{m=1}^M \left\{0<\max_t\left\|\bb E\tilde{A}_{s_m,e_m}^t\right\|_{op}<b\right\}\bigg|\mathcal{A}\right)\to 0.
	\end{align*}
	
	\textbf{Proof of $\bb P(\mathcal{A})\ge 1-T^{-1}$.}  It has been proved in (\ref{eqn:ea}) in the proof of Theorem~\ref*{th:K}.
	
	\textbf{Proof of (1-3) in Theorem~\ref{the:gmm}.} On event $\mathcal{A}$, we have $\left|
	||\tilde A_{s,e}^t||_{op}-||\bb{E}\tilde A_{s,e}^t||_{op} 
	\right|\le  C_3(\sqrt{\log(T)}/(1-\sqrt{\Xi})+\sqrt{n}/\sqrt{\log^*(1/\Xi)})$  for all $s<t<e, (s,e]\subset (0,T]$. Moreover, by Assumption \ref{ass:gmm} we can choose a $b$ such that  $C_3(\sqrt{\log(T)}/(1-\sqrt{\Xi})+\sqrt{n}/\sqrt{\log^*(1/\Xi)})=o(b)$. Then the results follow. Detailedly, taking the proof of (3) as an example, we derive that on event $
	\mathcal{A}$, 
	\begin{align*}
		&\bigcup_{m=1}^M \left\{C_3(\sqrt{\log(T)}/(1-\sqrt{\Xi})+\sqrt{n}/\sqrt{\log^*(1/\Xi)})
		<\max_t|| \tilde{A}_{s_m,e_m}^t||_{op}< b/2\right\}\\
		&\subset \bigcup_{m=1}^M \left\{0<\max_t|| \tilde{A}_{s_m,e_m}^t||_{op}<b\right\}.
	\end{align*}
	Therefore, 
	\[\bb P\left(\bigcup_{m=1}^M \left\{C_3(\sqrt{\log(T)}/(1-\sqrt{\Xi})+\sqrt{n}/\sqrt{\log^*(1/\Xi)})
	<\max_t|| \tilde{A}_{s_m,e_m}^t||_{op}< b/2\right\}\bigg|\mathcal{A}\right)\to 0.\]
\end{proof}

\section{Additional lemmas}
\label{suppsec:Additional lemmas}
\begin{lemma}
	\label{th:alg correctness}
	For any set $S=\{[u_m,v_m]\}_{j=1}^{M^*}$ where $ M^*$ is an integer, applying line 8-22 in Algorithm~\ref{alg:1} to $S$, then the output $\hat K$ and $S^{*}$ satisfy $l_j<r_{j}$ for $1\le j\le \hat K$.
\end{lemma}

\begin{proof}[Proof of Lemma~\ref{th:alg correctness}]
	We show by contradiction. Assume that $r_{j}\le l_j$ for some $j$. By construction, there are two disjoint intervals in $[1,r_2]$ because we drop the intervals whose left endpoints are less than $r_1$ in line 12. Therefore, we have $l_2>r_1$. Similarly, there are $j$ disjoint intervals in $[1,r_{j}]$ and $\hat K+1-j$ disjoint intervals in $[l_j,T]$. This leads to a fact that the number of disjoint intervals in $[1,T]$ is $\hat K+1$, which is impossible since $\hat K$ is the maximum number of disjoint intervals in $[1,T]$. 
\end{proof}

\begin{lemma}
	\label{lem:lemma7}
	Define the event
	\begin{equation}
		\label{eqn:M}
		\mathcal{M}=\bigcap_{k=1}^K \left\{\exists (s_m,e_m)\ s.t.\ \eta_k-\frac{1}{4}\Delta < s_m < \eta_k-\frac{1}{8}\Delta,\ \eta_k+\frac{1}{8}\Delta < e_m < \eta_k+\frac{1}{4}\Delta \right\}.
	\end{equation}
	We have$$
	\mathbb{P}(\mathcal{M})\ge 1-\frac{T}{\Delta}\exp\left(-\frac{M\Delta^2}{32T^2} \right).$$
\end{lemma}
\begin{proof}
	\begin{align}
		\mathbb{P}(\mathcal{M}^c)\le & \sum_{k=1}^K \prod_{m=1}^M \left( 1- \mathbb{P}\left(\eta_k-\frac{1}{4}\Delta < s_m< \eta_k-\frac{1}{8}\Delta,\ \eta_k+\frac{1}{8}\Delta <  e_m < \eta_k+\frac{1}{4}\Delta \right) \right)\notag\\
		= & \sum_{k=1}^K \prod_{m=1}^M \left( 1- \mathbb{P}\left(\eta_k-\frac{1}{4}\Delta < s\wedge e < \eta_k-\frac{1}{8}\Delta,\ \eta_k+\frac{1}{8}\Delta < s\vee e < \eta_k+\frac{1}{4}\Delta \right) \right)\notag\\
		=& K\left(1-2\left(\frac{\Delta}{8T}\right)^2 \right)^M
		\le \frac{T}{\Delta}\exp\left(-\frac{M\Delta^2}{32T^2} \right),\notag
	\end{align}
	where $s,e$ are two i.i.d. random variables that follow the discrete uniform distribution on $\{1,\cdots,T\}$ and the last inequality follows from the fact that $2\left(\frac{\Delta}{8T}\right)^2<1/32$.
\end{proof}

\begin{lemma}
	\label{lemma:markov conditional independent}
	Assume that $\{X_i\}_{i=1}^{2n+1}, n \ge 1$ is a discrete Markov chain with state space $\mathcal{S}$. Then conditional on $X_1,X_3,\cdots,X_{2n+1}$, the sequence $X_2,X_4,\cdots,X_{2n}$ is independent.
\end{lemma}
\begin{proof}
	For any $x_1,x_2,\cdots,x_{2n+1}\in \mathcal{S}$,
	\begin{align*}
		& \mathbb{P}(X_2=x_2,X_4=x_4,\cdots,X_{2n}=x_{2n}|X_1=x_1,X_3=x_3,\cdots,X_{2n+1}=x_{2n+1})\\
		&=\frac{\mathbb{P}(X_{2n+1}= x_{2n+1}|X_{2n}=x_{2n})\cdots \mathbb{P}(X_{2}=x_2|X_{1}=x_1)\mathbb{P}(X_1=x_1)}{\mathbb{P}(X_{2n+1}=x_{2n+1}|X_{2n-1}=x_{2n-1})\cdots \mathbb{P}(X_{3}=x_3|X_{1}=x_1)\mathbb{P}(X_1=x_1)}.\\
		&   \mathbb{P}(X_2=x_2|X_1=x_1,X_3=x_3,\cdots,X_{2n+1}=x_{2n+1})\\
		&=\frac{\mathbb{P}(X_{2n+1}=x_{2n+1}|X_{2n-1}=x_{2n-1})\cdots \mathbb{P}(X_{5}=x_5|X_{3}=x_3)\mathbb{P}(X_{3}=x_3|X_{2}=x_2)\mathbb{P}(X_{2}=x_2|X_{1}=x_1)}{\mathbb{P}(X_{2n+1}|X_{2n-1})\cdots \mathbb{P}(X_{3}=x_3|X_{1}=x_1)}\\
		&=
		\frac{\mathbb{P}(X_{3}=x_3|X_{2}=x_2)\mathbb{P}(X_{2}=x_2|X_{1}=x_1)}{\mathbb{P}(X_{3}=x_3|X_{1}=x_1)}.\\
		& \mathbb{P}(X_4=x_4,\cdots,X_{2n}=x_{2n}|X_1=x_1,X_3=x_3,\cdots,X_{2n+1}=x_{2n+1})\\
		&=\frac{\mathbb{P}(X_{2n+1}=x_{2n+1}|X_{2n}=x_{2n})\cdots \mathbb{P}(X_{5}=x_5|X_{4}=x_4)\mathbb{P}(X_{4}=x_4|X_{3}=x_3)}{\mathbb{P}(X_{2n+1}=x_{2n+1}|X_{2n-1}=x_{2n-1})\cdots \mathbb{P}(X_{5}=x_5|X_{3}=x_3)}.
	\end{align*}
	Therefore, conditional on $X_1,X_3,\cdots,X_{2n+1}$, $X_2$ and $(X_4,\cdots,X_{2n})$ are independent. By induction we get the desired result. 
\end{proof}

\begin{lemma}
	\label{lemma:markov tv conditional}
	Assume that $X_1,X_2,X_3$ is a Markov chain taking values in $\{0,1\}$. Let\\ $\Xi^*=\max_{x\in\{0,1\},i=1,2}|\mathbb{P}(X_{i+1}=x|X_i=0)-\mathbb{P}(X_{i+1}=x|X_i=1)|$. Assume that $\Xi^*\le \frac{1}{2}, \Xi^*\le \mathbb{P}(X_i=1)\le 1-\Xi^*, i=1,2,3$, then for any $x_1,x_3\in\{0,1\}$, \[|\mathbb{P}(X_2=1|X_1,X_3)-\mathbb{P}(X_2=1)|\le 3\Xi^*.\]
\end{lemma}
\begin{proof}
	Let $\theta=\mathbb{P}(X_2=1)$. For any $x_1,x_2,x_3\in\{0,1\}$, 
	\begin{align*}
		&\mathbb{P}(X_2=1|X_1=x_1,X_3=x_3)\\
		&=\frac{\mathbb{P}(X_3=x_3|X_2=1)\mathbb{P}(X_2=1|X_1=x_1)}{\mathbb{P}(X_3=x_3|X_2=1)\mathbb{P}(X_2=1|X_1=x_1)+\mathbb{P}(X_3=x_3|X_2=0)\mathbb{P}(X_2=0|X_1=x_1)}\\
		&:=\frac{A}{A+B}.
	\end{align*}
	For $i=1,2$, 
	\begin{align*}
		&|\mathbb{P}(X_{i+1}=x_{i+1}|X_i=x_i)-\mathbb{P}(X_{i+1}=x_{i+1})|\\
		&=|\mathbb{P}(X_{i+1}=x_{i+1}|X_i=x_i)(1-\mathbb{P}(X_{i}=x_i))-\mathbb{P}(X_{i+1}=x_{i+1}|X_i=1-x_i)\mathbb{P}(X_i=1-x_i)|\\
		&\le \Xi^* \mathbb{P}(X_i=1-x_i).
	\end{align*}
	So we have 
	\[(\mathbb{P}(X_3=x_3)-\Xi^*(1-\theta))(\theta-\Xi^*\mathbb{P}(X_1=1-x_1))\le A\le (\mathbb{P}(X_3=x_3)+\Xi^*(1-\theta))(\theta+\Xi^*\mathbb{P}(X_1=1-x_1)),\]
	\[(\mathbb{P}(X_3=x_3)-\Xi^*\theta)(1-\theta-\Xi^*\mathbb{P}(X_1=1-x_1))\le B\le (\mathbb{P}(X_3=x_3)+\Xi^*\theta)(1-\theta+\Xi^*\mathbb{P}(X_1=1-x_1)).\]
	We first derive the upper bound of $A/(A+B)$. 
	\begin{align*}
		\frac{A}{A+B}&\le \frac{
			(\mathbb{P}(X_3=x_3)+\Xi^*(1-\theta))(\theta+\Xi^*\mathbb{P}(X_1=1-x_1))}{
			\mathbb{P}(X_3=x_3)+\Xi^{*2}\mathbb{P}(X_1=1-x_1)
		}\\
		&\le \theta+\Xi^*\frac{\theta(1-\theta)+(1-\Xi^*)(1-2\theta\Xi^{*})}{(1+\Xi^{*2})(1-\Xi^*)}\\
		&\le \theta+\Xi^*\left(\frac{2}{1+\Xi^{*2}}+\frac{\Xi^*}{(1+\Xi^{*2})(1-\Xi^*)}\right)<\theta+3\Xi^*
	\end{align*}
	Similarly, $\frac{A}{A+B}\ge \theta-\Xi^*\frac{2-\Xi^*}{(1+\Xi^{*2})(1-\Xi^*)}\ge \theta-3\Xi^*$. Then the proof is complete.  
\end{proof}

\begin{lemma}
	\label{lem:lem8}
	For change-points detection in the piecewise constant case, suppose $(s,e]\subset (0,T]$ is an interval and there exists a change-point $\eta\in (s,e]$ such that any other change-point $\eta^{'}$ and the endpoints $s,e$ satisfies $\min\{\eta-s,e-\eta,|\eta-\eta^{'}|\}\ge \gamma$. Then $$
	\max_t ||\mathbb{E}\tilde{A}_{s,e}^t||_{op} \ge \frac{1}{4\sqrt{e-s}}\kappa\gamma
	$$
\end{lemma}
\begin{proof}
	Without loss of generality, shifting the values of  $A(t) $ such that  $\sum_{t=s+1}^e\mathbb{E}A(t) = 0$.  Due to the triangle inequality and the fact that $\{\bb{E}A(t)\}$ is piecewise constant, we have 
	\begin{align*}
		&\max\left\{\left\|\sum_{i=\eta-1-\gamma}^{\eta-1}\mathbb{E}{A}(i)\right\|_{op},
		\left\|\sum_{i=\eta}^{\eta+\gamma-1}\mathbb{E}A(i)\right\|_{op}
		\right\}\ge \frac{1}{2}\left\|\sum_{i=\eta-1-\gamma}^{\eta-1}\mathbb{E}A(i)-\sum_{i=\eta}^{\eta+\gamma-1}\mathbb{E}A(i)\right\|_{op}\\
		&=\frac{\gamma}{2} \left\|\mathbb{E}A(\eta-1)- \mathbb{E}A(\eta)\right\|_{op}\ge \frac{\gamma}{2}\kappa.
	\end{align*}
	Therefore, by the triangle inequality we have
	\[\max\left\{\left\|\sum_{i=s+1}^{\eta-1-\gamma}\mathbb{E}A(i)\right\|_{op},
	\left\|\sum_{i=s+1}^{\eta-1}\mathbb{E}A(i)\right\|_{op},
	\left\|\sum_{i=s+1}^{\eta-1+\gamma}\mathbb{E}A(i)\right\|_{op}
	\right\}\ge\frac{\gamma}{4}\kappa.\]
	
	Without loss of generality, consider that 
	\[
	\left\|\sum_{i=s+1}^{\eta-1}\mathbb{E}A(i)\right\|_{op}   \ge \frac{1}{4}\kappa_q p\gamma,\]
	Then it follows from the property of the CUSUM statistic that $$
	\max_t ||\mathbb{E}\tilde{A}_{s,e}^t(X)||_{op} \ge \sqrt{\frac{e-s}{(e-\eta+1)(\eta-1-s)}}\frac{1}{4}\kappa \gamma\ge \frac{1}{4\sqrt{e-s}}\kappa \gamma.
	$$ 
\end{proof}

\end{document}